\documentclass{article}

\usepackage{subfigure}
\usepackage{caption}
\usepackage{xcolor}
\usepackage{multirow}
\usepackage{booktabs}
\usepackage{xfrac}
\usepackage{csquotes}
\usepackage{adjustbox}
\usepackage{algorithm}
\usepackage{algcompatible}
\usepackage{url}
\usepackage{setspace}
\usepackage{caption}

\usepackage{amsthm}
\usepackage[margin=1.3in]{geometry}

\newtheorem{problem}{Problem}
\newtheorem{definition}{Definition}
\newtheorem{lemma}{Lemma}

\newcommand{\revision}[1]{{\leavevmode\color{black}#1}}
\algnewcommand\INPUT{\item[\textbf{Input:}]}%
\algnewcommand\OUTPUT{\item[\textbf{Output:}]}%

\begin{document}

\title{The Hourglass Effect in Hierarchical Dependency Networks}
\author{
Kaeser M Sabrin, Constantine Dovrolis \\
School of Computer Science \\
Georgia Institute of Technology\\
\textit{kmsabrin@gatech.edu} and \textit{constantine@gatech.edu}\\
}

\maketitle

\begin{abstract}
Many hierarchically modular systems are structured in a way that resembles an hourglass. 
This ``hourglass effect'' means that the system generates many outputs from many inputs through 
a relatively small number of intermediate modules that are critical for the operation of the entire system,
referred to as the waist of the hourglass.
We investigate the hourglass effect in general, not necessarily layered, hierarchical dependency networks.
Our analysis focuses on the number of source-to-target dependency paths that traverse each vertex,
and it identifies the core of a dependency network as the smallest set of vertices that collectively 
cover almost all dependency paths.
We then examine if a given network exhibits the hourglass property or not, 
comparing its core size with a ``flat'' (i.e., non-hierarchical) network that preserves
the source dependencies of each target in the original network. 
As a possible explanation for the hourglass effect,
we propose the {\em Reuse Preference (RP)} model that captures the bias of new modules 
to reuse intermediate modules of similar complexity instead of connecting directly to sources or low complexity
modules.  
We have applied the proposed framework in a diverse set of dependency networks from technological, 
natural and information systems, showing that all these networks exhibit the general hourglass property
but to a varying degree and with different waist characteristics. 

\footnote{\textcolor{blue}{This is a revised version of the paper, \enquote{The hourglass effect in hierarchical dependency networks}, published in the journal of \textit{Network Science} 5.4 (2017): 490-528. First, a typo has been corrected in the network of Figure~\ref{fig:h-score}. Second, the model of \hyperref[sec:model]{Section 6} has been revised to address a corner case that could occur for large values of $\alpha$. The differences with respect to the published paper are  highlighted  in  blue in that Section. This change also required to update Figures~\ref{fig:synthetic-alpha}, \ref{fig:model-alpha-effect-1}, \ref{fig:model-alpha-effect-2}, and \ref{fig:tau-effect-model} (with no qualitative differences in the results). We are grateful to Ankit Srivastava for pointing out this corner case.}}
\end{abstract}

\noindent
{\em Keywords:} Modularity, Hierarchy, Evolvability, Robustness, Complexity, Centrality, Core-Periphery
Networks, Hourglass Networks, Bow-Tie Networks, Dependency Networks.


\section{Introduction}

Complex systems in the natural, technological and information worlds are often hierarchically modular
\cite{meunier2010modular,parnas1984modular,ravasz2002hierarchical,schilling2000toward}.  
A modular system consists of smaller sub-systems (modules) that, at least in the ideal case, can 
function independently of whether or how they are connected to other modules: 
each module receives inputs from the environment or from other modules to perform a certain function  
\cite{baldwin2000design,callebaut2005modularity,wagner2007road}.
Modular systems are often also hierarchical,
meaning that simpler modules are embedded in, or reused by, modules of higher 
complexity \cite{ravasz2003hierarchical,sales2007extracting,simon1991architecture,yu2006genomic}.
In the technological world, modularity and hierarchy are often viewed as essential principles
that provide benefits in terms of design effort (compared to ``flat'' or ``monolithic'' designs in which 
the entire system is a single module), 
development cost (design a module once, reuse it many times), 
and agility (upgrade, modify or replace modules without affecting the entire system)
\cite{huang1998modularity,fortuna2011evolution,myers2003software}.
In the natural world, the benefits of modularity and hierarchy are often viewed in terms of 
evolvability  (the ability to adapt and develop novel features can be accomplished with
minor modifications in how existing modules are interconnected) 
\cite{kashtan2005spontaneous,kashtan2007varying,lorenz2011emergence}
and robustness (the ability to maintain a certain function even when there are internal or external 
perturbations can be accomplished using available modules in different ways)
\cite{kirsten2011evolution,kitano2004biological,stelling2004robustness}.
In information sciences, hierarchical modularity can improve the stability, quality and speed of
organizational search tasks (such as product or strategy development) \cite{mihm2010hierarchical,valverde2007self}. 
Additionally, it has been shown that both modularity and hierarchy can emerge naturally 
as long as there is an underlying cost for the connections between different system units
\cite{clune2013evolutionary,mengistu2015evolutionary}. 

It has been observed across several disciplines that hierarchically modular 
systems are often structured in a way that resembles a bow-tie or hourglass 
(depending on whether that structure is viewed horizontally or vertically).
Informally, this means that the system generates many outputs from many inputs through a relatively small 
number of intermediate modules, referred to as the ``knot'' of the bow-tie or the ``waist'' of the 
hourglass.\footnote{The two terms, 
bow-tie and hourglass, have not been always viewed as synonymous in the network science literature. In particular, 
the term bow-tie has been applied even to networks for which the knot includes a large fraction of the
network's vertices. We discuss the differences between the two terms in Section~\ref{sec:rlwork}.  }  
This ``hourglass effect'' has been observed 
in embryogenesis \cite{casci2011development,quint2012transcriptomic},
in metabolism \cite{ma2003connectivity,tanaka2005highly,zhao2006hierarchical},
in immunology \cite{beutler2004inferences,oda2006comprehensive},
in signaling networks \cite{supper2009bowtiebuilder},
in vision and cognition \cite{quiroga2005invariant,riesenhuber1999hierarchical},
in deep neural networks \cite{hinton2006reducing},
in computer networking \cite{akhshabi2011evolution},
in manufacturing \cite{swaminathan1998modeling}, 
as well as in the context of general core-periphery 
complex networks \cite{csermely2013structure,holme2005core}.

\begin{figure}
	\centering
	\subfigure[Pyramid: few targets depend on many sources (or the opposite)]
    {
        \includegraphics[width=0.2975\linewidth]{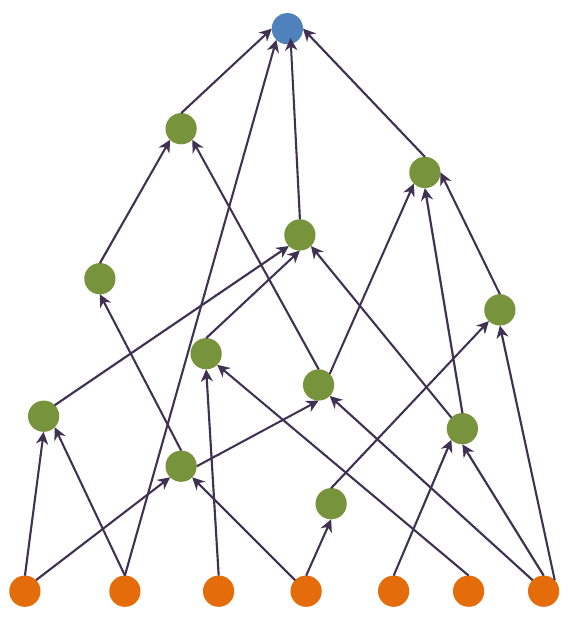}
    }
	\hspace{10mm}
    \subfigure[Direct: targets often depend directly on sources, few intermediate vertices]
    {
        \includegraphics[width=0.275\linewidth]{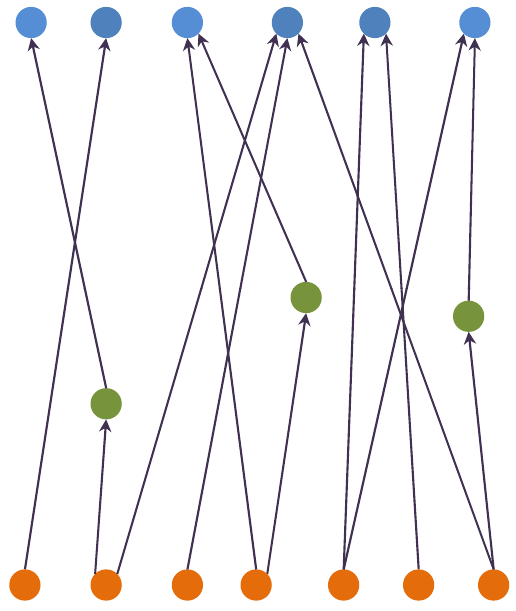}
    }
	\\
	\vspace{1mm}
	\subfigure[Decoupled: little reuse of common intermediate vertices across targets]
    {
        \includegraphics[width=0.275\linewidth]{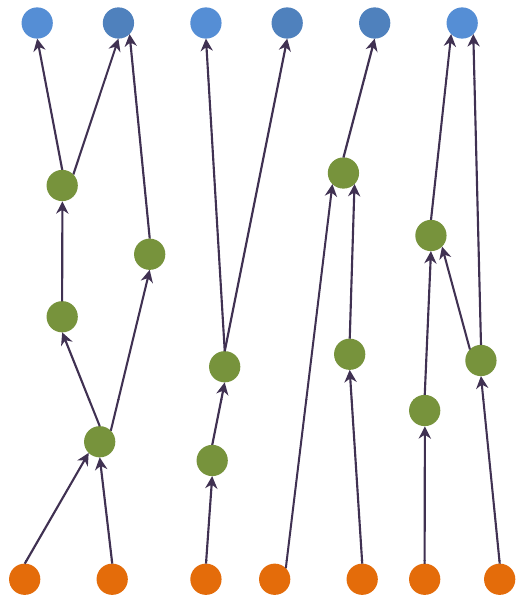}
    }
	\hspace{10mm}
	\subfigure[Hourglass: almost all source-target dependencies traverse a small number of intermediate vertices]
    {
        \includegraphics[width=0.275\linewidth]{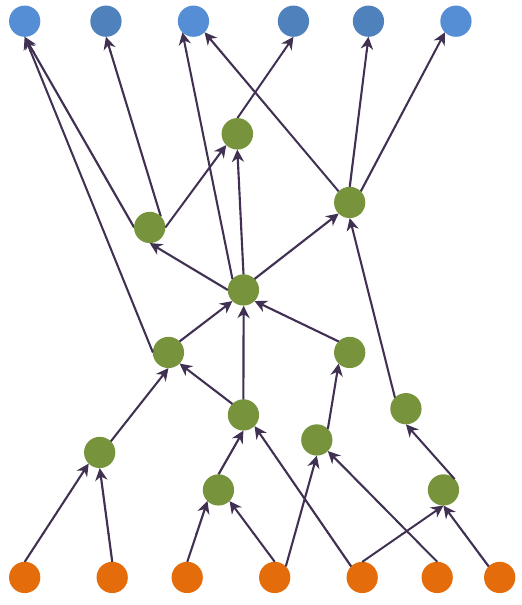}
    }
	\caption{Four toy examples of dependency networks with qualitatively different structure. The 
blue vertices are targets, the green are intermediates, and the orange are sources.}
  \label{fig:dependency-dag-shapes}
\end{figure}

The few intermediate modules in the hourglass waist 
are critical for the operation of the entire system,  and so they are also more conserved during the
evolution of the system compared to modules that are closer to inputs or outputs 
\cite{akhshabi2011evolution,csete2004bow,domazet2010phylogenetically}.
These observations have emerged in a wide range of natural, technological and information disciplines, 
and so it is interesting to investigate whether the so-called {\em hourglass effect} has 
deeper and more general roots that are largely domain-independent.

In this paper, we present a quantitative framework for the investigation of the hourglass effect based
on network analysis.  
First, the organization of a hierarchically modular system is transformed into a {\em dependency network},
i.e., a Directed Acyclic Graph (DAG) in which vertices represent either individual modules or Strongly 
Connected Components (SCCs) of interdependent modules. 
An edge from vertex $u$ to vertex $v$ in a dependency network 
means that module $v$ depends, in a domain-specific manner, on module $u$. 
The input vertices of the dependency network are referred to as {\em Sources} and the outputs as {\em Targets}.
For example, four dependency networks with different qualitative structures
are shown in Figure~\ref{fig:dependency-dag-shapes}.

The importance of each vertex is quantified with a {\em path centrality} metric, defined as the number  
of source-to-target dependency paths that traverse that vertex.  
Based on that metric,
we propose an algorithm to identify the {\em core} of the dependency network, i.e., the smallest set 
of vertices that collectively cover almost all (say 80-90\%) of all source-to-target dependency paths.
After computing the core, we can then 
evaluate if the given network exhibits the hourglass property or not by
comparing its core size with a ``flat'' (i.e., non-hierarchical) network that preserves
the source dependencies of each target.
We also present a {\em Reuse Preference (RP)} model for the formation of a dependency network,
capturing the bias of new modules
to reuse intermediate modules of similar complexity instead of connecting directly to sources or low complexity
modules.

We have applied this analysis framework in a diverse set of dependency networks from technological, 
natural and information systems: 
the call-graphs of two software systems, the metabolic networks of two species, and 
the citation networks of US Supreme Court cases for two legal matters (legality of abortion, 
and pension disputes).  
We show that these networks exhibit the hourglass property but to a varying degree.  
Further we quantify the location of the waist, relative to sources and targets, and the fraction of 
vertices in ``tendril'' paths that bypass the waist. 
The identified vertices at the waist of each network correspond to well-known important
modules in the corresponding systems. 

Finally, we discuss the connections between the hourglass effect and related concepts
such as the core-periphery structure of many complex networks, the presence of network bottlenecks, and the 
evolvability and robustness of systems that are hierarchically modular.  
Together with its theoretical significance, the hourglass effect may also have important practical value,
especially in the design of technological systems that operate in uncertain or evolving environments.

\revision{
In summary, the main contributions of this study are:\\  
1. To show how to transform a directed hierarchical network into a dependency network,
and to introduce path centrality as an appropriate metric for the analysis of dependency networks.\\ 
2. To formulate the core identification problem as 
finding the smallest set of vertices that are traversed by a given fraction $\tau$
of all source-target paths.\\
3. To show how to quantify whether a dependency network exhibits the hourglass effect.\\
4. To propose a probabilistic ``Reuse-Preference'' model of dependency network formation,
which illustrates the conditions under which a dependency network exhibits the hourglass effect.\\ 
5. To apply this analysis and modeling framework on several dependency networks from different
disciplines, showing that they all exhibit the hourglass effect but to a varying extent and with
different waist characteristics.\\
6. To discuss the significance of the hourglass effect in both technology and nature in terms
of network bottlenecks, cost, evolvability and robustness.
}

\section{Dependency networks} \label{sec:dependencies}
Suppose that we are given a directed network ${\bf G_0}$ that represents a hierarchically modular system. 
Each vertex of ${\bf G_0}$ corresponds to a system module.
An edge from vertex $u$ to vertex $v$ means that module $v$ {\em depends on} module $u$.
The precise meaning of this dependency relation is domain-specific.
In a software system, for instance, modules may represent C functions
and edges function calls (function $v$ calls $u$).
In a citation network, the modules may represent research papers or patents
and edges some form of information transfer ($v$ cites $u$).
In a mechanical or chemical process, the modules may represent different devices or materials
and the edges may represent that the construction (or composition)
of a device (material) $v$ requires $u$ as input.
Such hierarchical networks are ubiquitous across biology (e.g., food webs),
technology (e.g., communication protocol stacks), organizations (e.g., reporting
hierarchies), and information systems or social networks (e.g., meme propagation).

In general, the network ${\bf G_0}$ may include cyclic relations (referred to as ``feedback loops'', 
``recursive calls'', etc, depending on the context) between two or more vertices. 
Each set of such interdependent modules can be identified as a Strongly Connected Component
(an SCC is a set of vertices so that every vertex of that set can reach any other vertex of that set).
In other words, the modules of an SCC do not have any hierarchical ordering between them; they are all
mutually interdependent.
To construct an acyclic hierarchical network, 
we first compute all SCCs of ${\bf G_0}$;
this can be done in linear time using Tarjan's algorithm \cite{tarjan1972depth}.
Then, we replace every SCC of ${\bf G_0}$ with a single {\em super-vertex} that corresponds to the
set of vertices in that SCC. Any incoming edge to a vertex of an SCC from a vertex
outside that SCC becomes an incoming edge to the corresponding super-vertex;
similarly, we construct the outgoing edges of each super-vertex from the outgoing edges of the corresponding SCC. 
The replacement of SCCs with super-vertices transforms the original network ${\bf G_0}$ 
into a Directed and Acyclic Graph ${\bf G}$.
We refer to {\bf G} as the {\em dependency network} that corresponds to the original network ${\bf G_0}$.  

In the rest of the paper, the analysis will be focusing on dependency networks, 
and the notation will be as follows (Table~\ref{tab:list-of-symbols} in the Appendix lists all our notation).  
The {\em dependency network} ${\bf G}$ has 
a set ${\bf V}$ of vertices and a set ${\bf E}$ of directed edges.
The number of vertices and edges is denoted by $V$ and $E$, respectively.\footnote{We denote
the cardinality of a set ${\bf X}$ with $X$.}
The in-degree of $v$ is denoted by $d_{in}(v)$ and the set of 
vertices that point to $v$ is denoted by $I(v)$ ({\em inputs} of $v$).
Similarly, the out-degree of $v$ is denoted by $d_{out}(v)$ and the set of 
vertices that $v$ points to is denoted by $O(v)$ ({\em outputs} of $v$).
The {\em ancestors} of $v$ is the set of vertices that can reach $v$, 
while the {\em descendants} of $v$ is the set of vertices that $v$ can reach. 

The set ${\bf S}$ of vertices with zero in-degree are referred to as {\em Sources},
while the set ${\bf T}$ vertices with zero out-degree are referred to as {\em Targets}. 
The set ${\bf M}$ of remaining vertices represent 
{\em Intermediate} modules. We have that ${\bf V}={\bf S} \cup {\bf T} \cup {\bf M}$.
When plotting dependency networks, we follow the convention that sources 
appear at the bottom and targets at the top, and so edges have an upward direction. 

A path $p(s,t)$ from a source $s$ to a target $t$ is referred to as a {\em source-target path}, 
or simply {\em ST-path}.   
Focusing on a target $t$, the set of all ST-paths that terminate at $t$ 
represent the different ``dependency chains'' of sources and intermediates that 
are involved in the formation of $t$.  
We focus on all ST-paths that terminate at $t$ 
instead of all source and intermediate vertices that $t$ depends on.
This distinction is important because a source or intermediate vertex $v$ that participates
in several ST-paths that terminate at $t$ is more important for $t$ 
than a vertex $u$ that participates in fewer such ST-paths. 
For instance, in the context of a citation network the set of ST-paths that terminate
at $t$ represents all distinct ways in which the information contained in those
source references has been transformed and propagated by intermediate references to
finally produce $t$.  
 
To quantify the topological importance of a vertex in a dependency network we rely on the following metric:
\begin{definition}[Path Centrality]
The path centrality $P(v)$ of a vertex $v$ is the number of ST-paths that traverse $v$.
\end{definition}
This metric has been also referred to as the {\em stress} of a vertex \cite{ishakian2012framework}. 
Fig.\ref{fig:path-centrality-complexity-generality}-a illustrates the path centrality of each vertex
in a small dependency network.

$P(v)$ can be computed in $O(E)$ time, 
due to the acyclic nature of  dependency networks.
Suppose that $P_S(v)$ is the number of paths from any source to $v$,
while $P_T(v)$ is the number of paths from $v$ to any target.
$P_S(v)$ can be computed in a bottom-up manner:
$P_S(v)=1$ for all sources and $P_S(v)=\sum_{u \in I(v)} P_S(u)$ for any $v$ that is not a source.
Similarly, $P_T(v)$ can be computed in a top-down manner: 
$P_T(v)=1$ for all targets and $P_T(v)=\sum_{u \in O(v)} P_T(u)$ for any $v$ that is not a target.
It is easy to see that the path centrality of $v$ is simply the product of $P_S(v)$ and $P_T(v)$,
\begin{equation} \label{eq:path-centrality}
P(v) = P_S(v) \times P_T(v)
\end{equation}

The path centrality metric can be also interpreted as follows.  
The number of paths $P_S(v)$ from sources to $v$ can be thought of as a proxy for $v$'s {\em complexity}: 
The more ST-paths terminate at $v$, the more complex is the formation of $v$ from all its ancestors. 
Sources  have minimal complexity (set to one) because they do not depend on anything else. 
On the other hand,
the number of paths $P_T(v)$ from $v$ to targets can be thought of as a proxy for $v$'s {\em generality}: 
The more ST-paths exist from $v$ to the set of targets, the more general or common is the function
provided by $v$ in the formation of distinct targets. 
Targets have minimal generality (set to one) because they are not used to form any other module. 

Based on these two definitions, 
{\em the path centrality of a vertex $v$ is the product of $v$'s complexity and generality.}  
This implies that path centrality is a metric that evaluates the topological importance of a
vertex in both the upward and downward directions of a dependency network. If the 
complexity and generality of a vertex are both high, relative to other vertices, that 
vertex will also have high path centrality.  
Fig.\ref{fig:path-centrality-complexity-generality}-b illustrates the complexity and generality 
of each vertex in a small dependency network.

The path centrality metric is more appropriate for identifying important vertices in a dependency network 
than other centrality metrics.
The betweenness or closeness centrality metrics, for instance, only consider the shortest paths between
two vertices, and so they would not assign high centrality to a vertex that participates in 
many (but relatively long) ST-paths.
Also, the in-degree or out-degree of a vertex is a local metric and it does not capture the positioning 
of that vertex in the entire dependency network. 
The Katz centrality metric, on the other hand, does not distinguish between intermediate vertices
and terminal (source or target) vertices, and it penalizes longer dependency paths. 
Some other centrality metrics, such as pagerank or eigenvector centrality \cite{newman2010networks}, 
are not appropriate for DAGs. 

\begin{figure}
	\centering
	\subfigure[Path centrality]
    {
        \includegraphics[scale=0.415]{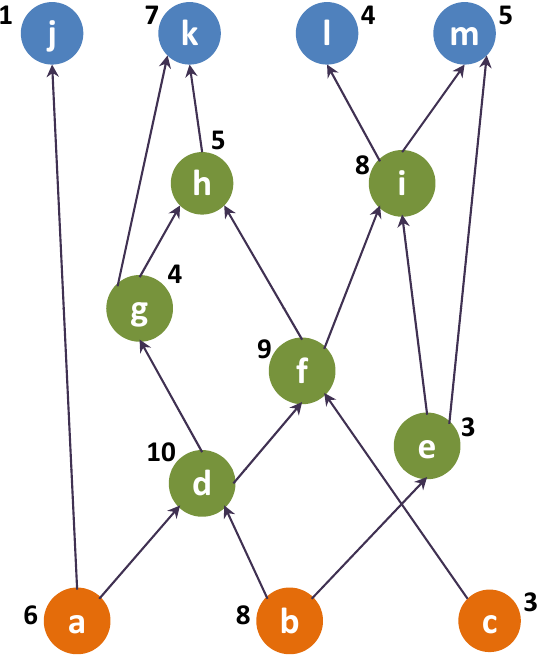}
    }
	\hspace{6mm} 
    \subfigure[Complexity \& generality]
    {
        \includegraphics[scale=0.415]{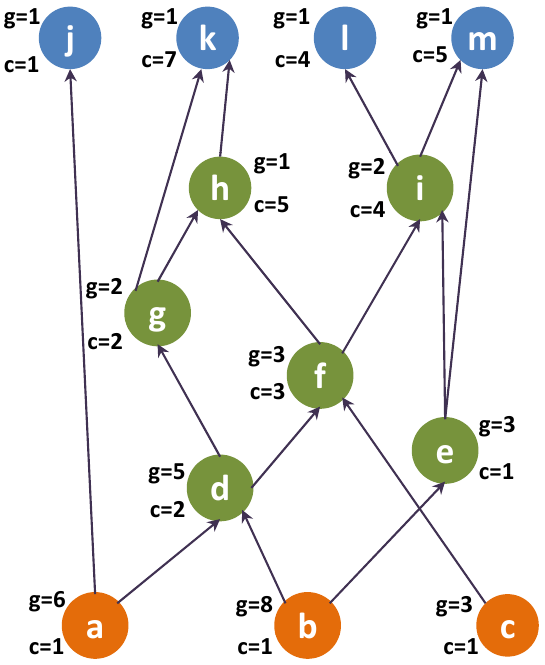}
    }
	\caption{The path centrality of each vertex (shown at the left) and the generality (top number) and complexity 
(bottom number) of each vertex (shown at the right).} 
    \label{fig:path-centrality-complexity-generality}
\end{figure}

\section{The core of a dependency network} \label{sec:core} 
Intuitively, the {\em core} of a dependency network can be defined 
as a subset of vertices that represent the most central modules in the underlying system.
One approach would be to rank vertices in terms of path centrality. 
This approach does not consider however that two or more vertices may be traversed by  
almost the same set of ST-paths. So, even though they may both have high path centrality, including 
one of them in the core would be sufficient to ``cover'' those source-target dependencies.   

Instead, we define the core of a dependency network based on the solution of an optimization 
problem: identify the smallest set of vertices that are traversed by almost all ST-paths -- namely, a
large fraction $\tau$ of all ST-paths.
We approach this problem in two steps. First, we consider the problem of computing the most central set 
of $k$ vertices, when $k$ is given, which has already been studied by Ishakian et al. in 
\cite{ishakian2012framework}. Then, we use an algorithm for the previous problem to
identify the minimum-size core for a given fraction $\tau$ of ST-paths.

\begin{definition}[Coverage of ST-paths]
Let ${\bf P}$ be the set of all ST-paths and ${\bf R}$ be a set of vertices. 
${\bf P}_{\bf R}$ is the subset of ${\bf P}$ that traverses at least one vertex in ${\bf R}$.   
The corresponding {\em path coverage} of ${\bf R}$ is defined as: 
\begin{equation}
\delta_{\bf R} = \frac{P_R}{P}
\end{equation} 
\end{definition}

\begin{problem}[Cardinality-Constrained Core with Maximum Coverage --  {\bf C$^3$MC} ]
Given a cardinality $k$, identify a set ${\bf \hat{R}}_k$ of $k$ vertices with maximum path coverage. 
\begin{equation}
{\bf \hat{R}}_k = {\arg\max}_{{\bf R} \subset {\bf V}: |{\bf R}|=k} \, \{{\delta}_{\bf R}\}
\end{equation}  
The set ${\bf \hat{R}}_k$ may not be unique but $\delta_{{\bf R}_k}$, denoted as $\hat{\delta}_k$ in the following, 
is the same for all optimal solutions. 
\end{problem}

The {\bf C$^3$MC} problem is NP-Complete; a proof is given by Ishakian et al. \cite{ishakian2012framework}.
However, the objective function of the {\bf C$^3$MC} problem is monotonically increasing (obvious) and submodular 
(proven in the Appendix), and so  
the following greedy algorithm is guaranteed to produce an $(1-\frac{1}{e})$-approximation of the optimal solution
\cite{nemhauser1978analysis} -- the same algorithm was also used in the work of Ishakian et al.\\
\begin{itemize}
\item Initially, the set ${\bf \hat{R}}_k$ is empty.
\item In each iteration:
\begin{enumerate}
\item Compute the path centrality of all vertices. 
\item Include the vertex with maximum path centrality in the set ${\bf \hat{R}}_k$, 
and remove it from the network (the case of ties is discussed in \S~\ref{sec:ties}). 
\end{enumerate}
\item The algorithm terminates when the set ${\bf \hat{R}}_k$ includes $k$ vertices. 
\end{itemize}

The run-time complexity of the path centrality computation is $O(E)$ 
and, in the worst case, we need to recompute the path centrality of all vertices in every iteration of the algorithm. 
So, the run-time complexity of the previous greedy algorithm is $O(k \, E)$. 
In Section~\ref{sec:runtime}, we show experimentally that the run-time of the core
identification algorithm increases quadratically with the number of vertices $N$
(if the average in-degree of non-source vertices remains constant).

\subsection{Path centrality ties} \label{sec:ties}
We now describe how the previous greedy algorithm breaks ties among vertices that have the 
same maximum path centrality.
Figure~\ref{fig:node-disjunction} illustrates that there are two different types of ties.
First, it may happen that the tied vertices are traversed by exactly the same set of ST-paths.
This will be the case, for instance, when those vertices are connected in a linear chain (vertices $a$, $b$ and $c$
in Figure~\ref{fig:node-disjunction}).  
Whenever there is a maximum path centrality tie among a set of vertices that are traversed 
by the same set of ST-paths,
we group these vertices as a single {\em Path-Equivalent Set} (PES). 
The elements of a PES are equivalent in the sense that they all capture the same set of ST-paths;
it is sufficient to include any one of them in the core. 

\revision{To identify a 
PES from a set of vertices that have equal path centrality, we pick a vertex $u$ from that set and 
remove it from the network. Then, we recompute the path centrality of the remaining tied vertices and find those
that now have zero path centrality. These vertices, together with $u$, form a PES. 
We repeat this process for any remaining tied vertices to identify additional PESs.} 

\begin{figure}
    \centering
		\includegraphics[scale=0.425]{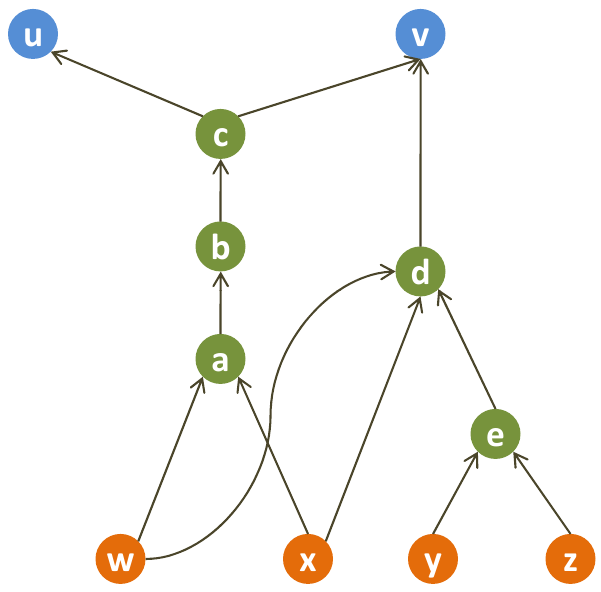}
    \caption{Vertices ${a,b,c}$ have equal path centrality and they are traversed by the same set of ST-paths. Vertex $d$ has equal path centrality but it is traversed by a different set of ST-paths. If there is a tie between the vertices ${a,b,c,d}$ in an iteration of the core identification algorithm, 
either the first three vertices will be added in the core as a single Path-Equivalent Set (PES) 
representing the set $\{a,b,c\}$, or only vertex $d$ will be added in the core.}
    \label{fig:node-disjunction}
\end{figure}

Second, there may be a maximum path centrality tie between two or more vertices that are 
traversed by different sets of ST-paths (for instance, vertices $a$ and $d$ in Figure~\ref{fig:node-disjunction}).  
Ties of this type can be randomly broken, 
as long as it is sufficient to identify a single core instead of enumerating all possible cores.
If it is necessary to identify all possible cores, 
we can consider separately every possible tie-breaker. This creates a tree of 
possible execution paths in which each leaf corresponds to a candidate core with $k$ elements.

\subsection{The path coverage threshold $\tau$}
In practice, the cardinality of the core is not known a priori. 
Instead, we can set the cardinality of the core heuristically, as follows. 

If it was required that the core is traversed by {\em all} ST-paths, 
the identification of the core would be equivalent to the well-known minimum-cut problem
that can be solved efficiently with a max-flow algorithm \cite{ravindra1993network}. 
However, requiring that the core covers every single ST-path is a very stringent condition;
we have observed that in real dependency networks 
there are often some direct ST-paths that do not traverse any intermediate vertices or that do not share common 
intermediate vertices with most other ST-paths. 

So, a more pragmatic definition is that the core of a dependency network 
should cover at least a fraction $\tau$ of all ST-paths,
where $\tau$ is a given {\em path coverage threshold} that will typically be close to one.
To compute the core, we solve the {\bf C$^3$MC} problem iteratively, starting with $k$=1.
The set ${\bf \hat{R}}_k$ is computed incrementally by adding one more vertex in ${\bf \hat{R}}_{k-1}$,
which requires only $O(E)$ additional operations.
The algorithm terminates when the path coverage $\hat{\delta}_k$ first exceeds $\tau$.

We use the following notation to represent the core of a dependency network for a given $\tau$:
the set of vertices in the core is ${\bf C(\tau)}$, 
the size of the core is $C(\tau)$,
and the path coverage of the core is $\delta_{{\bf C}(\tau)} \geq \tau$. 
Note that ${\bf C(\tau)}$ and $\delta_{{\bf C}(\tau)}$ may not be unique if there were ties 
during the computation of the core.
The core size $C(\tau)$, however, is unique. 

\revision{
The incremental increase of the path coverage of ${\bf C}$ when $v$ is first included in that core
is denoted by $\delta_{{\bf C}(v)}$. This metric also represents the {\em weight} of $v$ in the core.  
}

\section{Hourglass dependency networks} \label{sec:hourglass}

\subsection{Network flattening and H-score}
Informally, the hourglass property of a dependency network can be defined as having a small core,
even when the path coverage threshold $\tau$ is close to one.  
To make the previous definition more precise, we can compare the core size $C(\tau)$ of the given dependency
network ${\bf G}$ with the core size of a derived dependency network that maintains the same 
set source-target dependencies 
of ${\bf G}$ but that is not an hourglass by construction. 

To do so, we create a {\em flat dependency network} ${\bf G_f}$ from ${\bf G}$ as follows:
\begin{enumerate}
\item ${\bf G_f}$ has the same set of source and target vertices as ${\bf G}$ but it does not have any intermediate vertices.
\item For every ST-path from a source $s$ to a target $t$ in ${\bf G}$, we add a direct edge from $s$ to $t$ in ${\bf G_f}$. If there are $w$ edges from $s$ to $t$ in ${\bf G_f}$, they can be replaced with a single edge of weight $w$. 
\end{enumerate}
Note that ${\bf G_f}$ preserves the source-target dependencies of ${\bf G}$: 
each target in ${\bf G_f}$ is constructed based on the 
same set of ``source ingredients'' as in ${\bf G}$. 
Additionally, the number of ST-paths in the original dependency network is equal 
to the number of paths in the weighted flat network (an edge of weight $w$ counts as $w$ paths). 
However, the dependency paths in ${\bf G_f}$ are direct, without forming any intermediate modules that could be 
reused across different targets. 
So, by construction, the flat network ${\bf G_f}$ cannot have the hourglass property.

Suppose that $C_f(\tau)$ represents the core size of the flat network ${\bf G_f}$. 
The core of ${\bf G_f}$ can include a combination
of sources and targets, and it cannot be larger than either the set of sources or targets.
Additionally, the core of the flat network is larger or equal than the core of 
the original network (because the core of the flat network also covers at least a fraction $\tau$ of the ST-paths
of the original network -- but the core of the original network may be smaller because it can also include
intermediate vertices together with sources or targets).
So,
\begin{equation}
C(\tau) \leq C_f(\tau) \leq \min\{S, T\}
\end{equation} 

To quantify the extent at which ${\bf G}$ exhibits the hourglass effect,
we define the {\em Hourglass Score}, or {\em H-score}, as follows: 
\begin{equation}
H(\tau) = 1 - \frac{C(\tau)}{C_f(\tau)} 
\end{equation}
Clearly, $0\leq H(\tau) < 1$. 
The H-score of $G$ is approximately one if the core size of the original network is negligible compared to
the the core size of the corresponding flat network.
Figure~\ref{fig:h-score} illustrates the definition of this metric. 

\begin{figure}
	\centering
    {
        \includegraphics[scale=0.4]{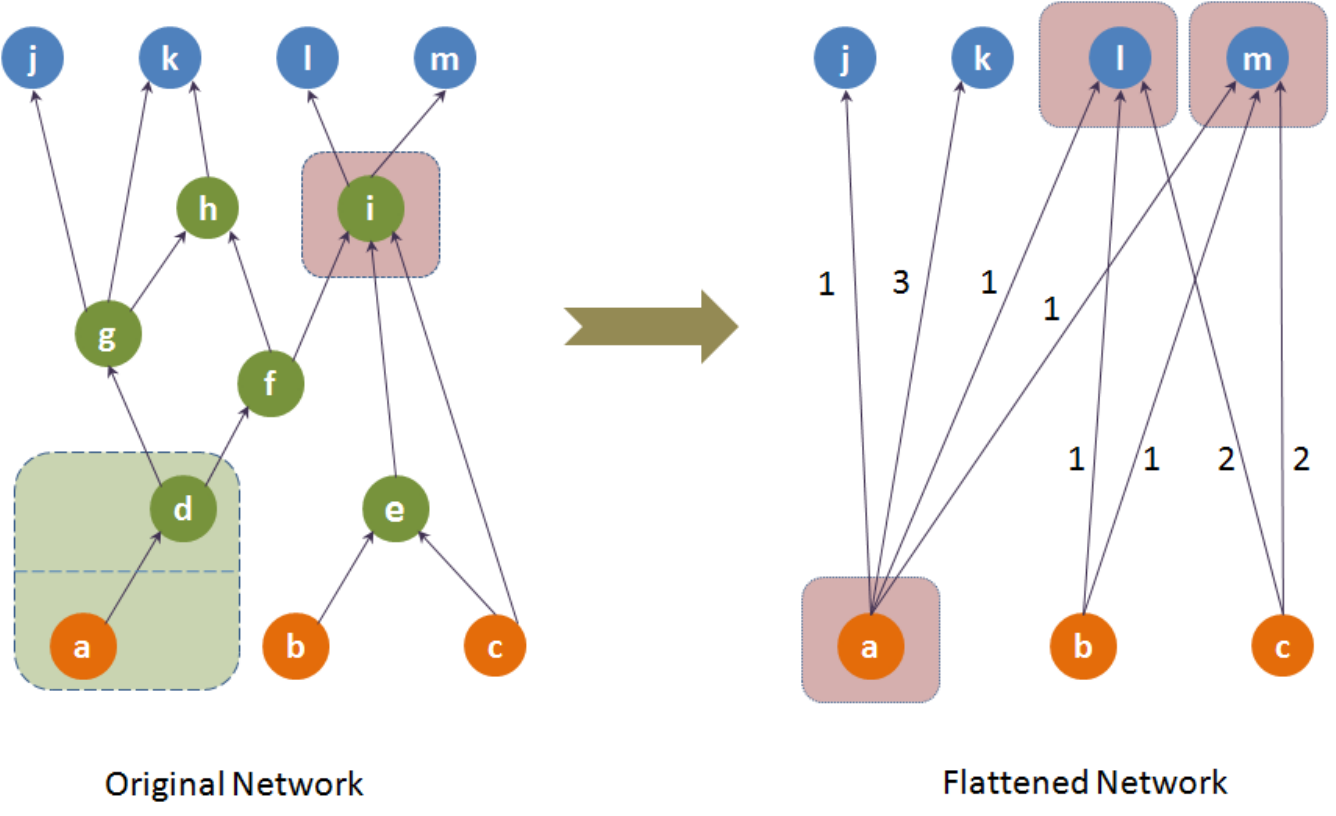}
		}
	\caption{The weight of an edge in the flattened network represents the number of ST-paths 
between the corresponding source-target pair in the original dependency network.
When the path coverage threshold is $\tau$=90\%, the core of the original network (left) is 
the set $\{\{a,d\},i\}$ ($\{a,d\}$ form a \textit{Path-Equivalent Set} and only one of them should be 
included in the core). The core of the flattened network (right) for the same $\tau$ is $\{a,l,m\}$. 
The H-score of the original network is $1-\frac{2}{3} = 0.33$.} 
  \label{fig:h-score}
\end{figure}

An ideal hourglass-like network would have a single intermediate vertex that is traversed by 
every single ST-path (i.e., $C(1)$=1), and a large number of sources and targets none of which 
originates or terminates, respectively, a large fraction of ST-paths (i.e., a large value of $C_f(1)$). 
The H-score of this network would be approximately equal to one.

\subsection{Coverage and location of a vertex}
Another property of an ideal hourglass network is that all vertices that participate in 
ST-paths should be reachable from the waist, either in the upstream or in the downstream direction.  
To quantify this property, we define the {\em core vertex coverage} metric $U_{\bf C}$, where ${\bf C}$ is the core
of the given dependency network: 
\begin{equation} 
U_{\bf C} = \frac{\sum_{v\in {\bf V_{ST}}}\phi_{\bf C}(v)}{V_{ST}}
\end{equation}
where ${\bf V_{ST}}$ is the set of vertices that are present in one or more ST-paths,
and
$\phi_{\bf C}(v)$ is equal to one when $v$ is a vertex that can reach, or that can be reached from, 
at least one vertex in the core ${\bf C}$; $\phi_{\bf C}(v)$ is zero otherwise. 
The metric $1-U_{\bf C}$ can be thought of as the fraction of
vertices in ``tendril'' paths that bypass the waist.


\revision{
We can also associate a {\em location} with each vertex to capture its relative position in the dependency 
network between sources and targets.
Computing the location of a vertex based on the {\em topological sorting} of the depending network
would not be an appropriate approach in this context 
because that ordering is determined from the maximum distance of a vertex from the set of sources.   
Another way to place intermediate vertices between sources and targets is to consider the 
complexity $P_S(v)$ and generality $P_T(v)$ metrics that were defined in Section~\ref{sec:dependencies}. 
Recall that sources have the lowest complexity value (equal to $1$), while targets have the lowest generality 
value (equal to $1$).
The following equation defines a location metric based on $P_S(v)$ and $P_T(v)$,
\begin{equation}
L(v) = \frac{P_S(v) - 1}{(P_S(v) - 1) + (P_T(v) - 1)}
\end{equation}
$L(v)$ varies between 0 (for sources) and 1 (for targets). 
If there is a small number of paths from sources to a vertex $v$ (low complexity) 
but a large number of paths from $v$ to targets (high complexity), $v$'s role in the network is more similar
to sources than targets, and so its location should be closer to 0 than 1. 
The opposite is true for vertices that have high complexity but low generality -- their location should be
closer to 1 than 0.
}

We can also calculate an {\em average location for the entire core}.
The weight of a core vertex $v$ is proportional to the incremental increase $\delta_{{\bf C}(v)}$ 
of the path coverage of ${\bf C}$ when $v$ was first included in that core.
So, the average location of the core ${\bf C}$ can be defined as 
the following weighted average of the location of the core vertices,
\begin{equation}
L_{\bf C} = \frac{\sum_{v \in {{\bf C}}} [\delta_{{\bf C}(v)} \, L(v)]}{\sum_{v \in {{\bf C}}} \delta_{{\bf C}(v)}} 
\end{equation}

\section{Case studies} \label{sec:realnets}
In this section, we apply the previous analysis framework in six 
dependency networks from three different disciplines:
two call-graphs (software engineering), two metabolic networks (biology, biochemistry) and two citation
networks (information science). 
First, we present the corresponding datasets and the process to convert them into dependency networks.  
Table~\ref{tab:data-networks} shows the basic characteristics of the six dependency networks.
Note that the networks vary considerably in terms of density, fraction of source or target vertices, and average
ST-path length.

\subsection{Datasets and dependency network construction} 
\subsubsection{Call-graphs}
Any non-trivial software system is written in a modular and hierarchical manner: ``functions'' (or ``methods'')
are defined for distinct processing of tasks,
and a function performs its task by calling other, simpler functions.
The resulting hierarchy of function calls is referred to as the {\em call-graph} of that system.
The sources of a software system are elementary functions that do not call any other function,
functions provided by linked libraries, 
or functions that communicate directly with the primitives provided by the underlying hardware
(e.g., device drivers) or the operating system.
The targets are various applications or utilities that are called by external entities
(the human user, other applications, libraries, systems, etc).

In the following, we analyze the call-graph of two complex and popular software systems:
OpenSSH (version 5.2, written in C) and the Apache Math library (version 3.4, written in Java).
The source code for OpenSSH was retrieved in a curated form from an earlier study \cite{bhattacharya2012graph}
and the call-graph was constructed using CodeViz \cite{codeviz}.
For the Apache Math library, we use the Java dependency graph extraction tool \cite{java-callgraph}.
We follow the earlier convention that when a function $v$ calls a function $u$, there is an edge from $u$ to $v$.

In the case of OpenSSH, we first remove from the call-graph all functions that include 
the following keywords in their name: 
{\em main} (included in many C files for testing different parts of the system independently),
{\em log} and {\em debug} (used during software development for debugging),
{\em fail}, {\em fatal}, {\em error} (generic functions called in case of unexpected errors),
and {\em exit} (program termination). 
The previous functions have high path centrality mostly because they are called by many other functions
but they do not provide any information about the system architecture.  
Similarly, for the Apache Math library, we remove all {\em exception handlers} (methods associated
with unexpected errors) and the methods of the {\em Object} class, which is the generic 
parent class for all Java programs.

The use of {\em recursive programming} (i.e., one or more functions forming a loop in the call-graph) creates
cycles. As discussed in Section~\ref{sec:dependencies}, each call-graph is transformed into a dependency network
by first partitioning the call-graph in a set of SCCs, and then replacing each SCC with a single {\em super-vertex}.
The number and size of the super-vertices in each call-graph are shown in Table~\ref{tab:data-networks}.

\begin{table}
\scriptsize
\begin{adjustbox}{width=1\textwidth}
\begin{tabular*}{\textwidth}{l @{\extracolsep{\fill}} cccccc}
		\multirow{5}{*}{Properties} & \multicolumn{6}{c}{Networks} \\ 
		\cline{2-7}  
		& \multicolumn{2}{c}{Software Call-graphs} & \multicolumn{2}{c}{Metabolic Networks} & \multicolumn{2}{c}{SCotUS Citation Networks} \\ 
		& \textit{OpenSSH} & \textit{Apache Math} & \textit{Rat} & \textit{Monkey} & \textit{Abortion} & \textit{Pension} \\
		& \textit{v-5.2} & \textit{v-3.4} & & & \textit{Cases} & \textit{Cases} \\
		\hline \hline
		Vertices                     & 1300 			 & 6685  			 & 843  			 & 845  			 & 1502  & 1290 \\
		Largest component (L-WCC)    & 99\% 			 & 95\%  			 & 64\% 			 & 61\% 			 & 100\% & 95\%\\
		Edges 				               & 4583 			 & 14823 			 & 612  			 & 588  			 & 3266  & 1555 \\
		Average degree               & 3.5 			   & 2.3  			 & 1.2 			   & 1.1 			   & 2.2   & 1.3 \\
		Targets                      & 22\% 			 & 35\%  			 & 24\% 			 & 25\% 			 & 20\%  & 24\% \\
		Intermediates                & 45\% 			 & 32\%  			 & 56\% 			 & 55\% 			 & 17\%  & 11\% \\
		Sources                      & 33\% 			 & 33\%  			 & 20\% 			 & 20\% 			 & 63\%  & 65\% \\
		Average ST-path length       & 10.4        & 8.8  			 & 8.3  			 & 8.1  		   & 14.1  & 5.1 \\
		Number of super-vertices     & 3           & 24    			 & 10   			 & 9    			 & 0     & 0\\
		Super-vertex size            & $2.5\pm0.5$ & $3.2\pm4.1$ & $9.4\pm7.4$ & $9.3\pm7.2$ & - 		 & - \\
		\hline
\end{tabular*}
\end{adjustbox}
\caption{Basic characteristics of analyzed dependency networks. All entries after the first row correspond to the 
Largest Weakly Connected Component (L-WCC).}
\label{tab:data-networks}
\end{table}

\subsubsection{Metabolic networks}
Metabolic networks show how individual chemical reactions in the cell are  
combined to form the complex pathways associated with functions such as glycolysis or the
biosynthesis of pyrimidine or purine \cite{palsson2015systems}. 
There are large databases that provide reasonably accurate and complete
metabolic networks for many species \cite{kanehisa2000kegg}. 
The KEGG database, in particular, has been curated for more than a decade to include all known metabolic reactions
that conform with the available sequenced genome information \cite{kanehisa2014data}.

In a metabolic network, the products of one chemical reaction can be used as substrates for
another chemical reaction. This flow of matter and energy can be represented as a directed network
where vertices correspond to metabolites, and an edge from $u$ to $v$
means that there is at least one reaction in which $u$ is a substrate (input) and $v$ is a product (output).
Although most chemical reactions are reversible, most metabolic pathways are
typically considered to flow in one direction.
In the KEGG database, each reaction is associated with the most common direction in a given pathway.

\begin{figure}
    \centering
                \includegraphics[scale=0.5]{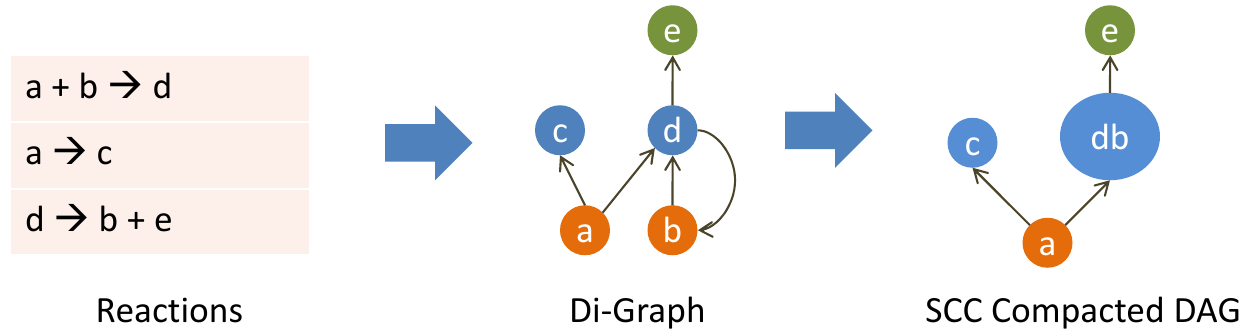}
    \caption{Construction of a dependency network from a given set of metabolic reactions.}
    \label{fig:metabolic-network-construction}
\end{figure}
A metabolic network often includes cycles.
If two or more metabolites are present in the same cycle, it means that there is no hierarchical
ordering between them -- they are mutually interdependent.  
So, as in the case of call-graphs, after constructing the initial metabolic network we replace each
SCC with a single {\em super-vertex} that represents the corresponding set of metabolites in that SCC.
Figure~\ref{fig:metabolic-network-construction} shows a small example of how a given set of chemical reactions
can be first transformed to a directed network, and then to a dependency network.

In the following, we present results for the metabolic networks of two organisms:
\textit{Rattus norvegius} (rat) and \textit{Macaca mulatta} (monkey).
Both datsets were retrieved from the 2014 KEGG \cite{kanehisa2014data} database.
For each metabolic network we only analyze the Largest Weakly Connected Component (L-WCC).
The smaller connected components correspond to distinct pathways that do not have any
common metabolites with the L-WCC.

\subsubsection{SCotUS citation network}
Dependency networks can also capture the flow of information, knowledge or legal precedent in research
publications, patents, court cases, and so on.
Here, we focus on the citation network of court judgments made by the Supreme Court of the United States (SCotUS).
We rely on a dataset collected by Fowler \cite{fowler2007network, fowler2008authority} that
includes all SCotUS cases between 1754 and 2002.
Judicial decisions often leverage the precedent of earlier judgments to support their arguments,
forming a directed citation network.
Following our earlier convention, if a court case $v$ refers to a previously settled case $u$,
there is an edge from $u$ to $v$.
In the case of citation networks, the hierarchy of the dependency network
implies a temporal ordering between connected vertices: if there is a path from $u$ to $v$, $u$ appeared before $v$.

In this paper, we focus on two legal matters that have been the subject of many SCotUS cases:
the {\em legality of abortion} and various {\em pension (or benefits) disputes}.
First, we use the Legal Information Institute \cite{cornell-law} of Cornell University's online legal library to find the set of
SCotUS cases that focus on each of these two  matters.
Suppose that ${\bf X}$ is the set of SCotUS cases that are related to one of these two matters.
We construct the corresponding citation network by including all cases in ${\bf X}$
as well as any other SCotUS case that directly cites, or is directly cited by, a case in ${\bf X}$.
This expansion of the citation network with cases that do not belong in ${\bf X}$ is important because the
SCotUS decisions about a certain matter may depend on, or they may have influenced, decisions regarding
other legal matters.

The selection of sources and targets in a citation network may appear as somewhat arbitrary.
This is an important issue that deserves further discussion.
The sources and targets of a dependency network should be selected based on the scope, or boundaries,
of the underlying system we aim to understand.
Considering only parts of that system, or merging it with other systems, can mislead the analysis.
For instance, if we want to identify the most significant publications associated with a specific
problem in network science, say community detection, it would be incorrect to only consider the
citation network of publications that focus on spectral graph partitioning, and it would also be incorrect
to consider every publication that relates broadly to graphs or networks.
We admit, however, that in some cases it may be challenging to uniquely identify the scope, or boundaries,
of a given dependency network; this is a problem that deserves further study.

The two citation networks are acyclic, and so we do not create any super-vertices.

\subsection{Analysis of dependency networks}

\begin{figure}
        \centering
        \subfigure[OpenSSH-v5.2]
    {
        \includegraphics[scale=0.31]{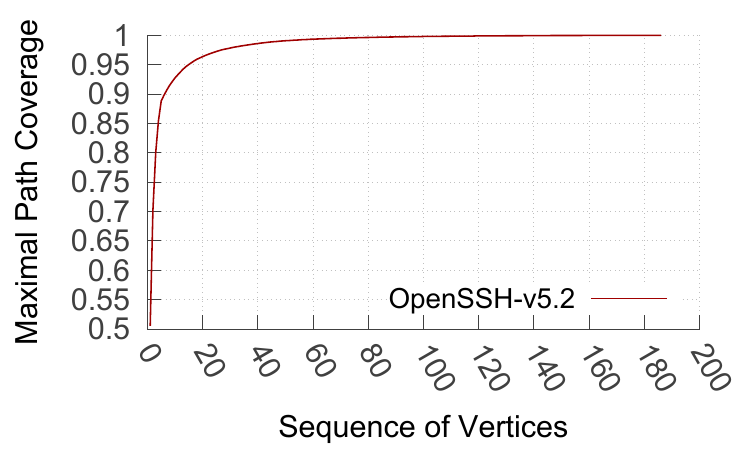}
    }
        \subfigure[Apache-Math-v3.4]
    {
        \includegraphics[scale=0.31]{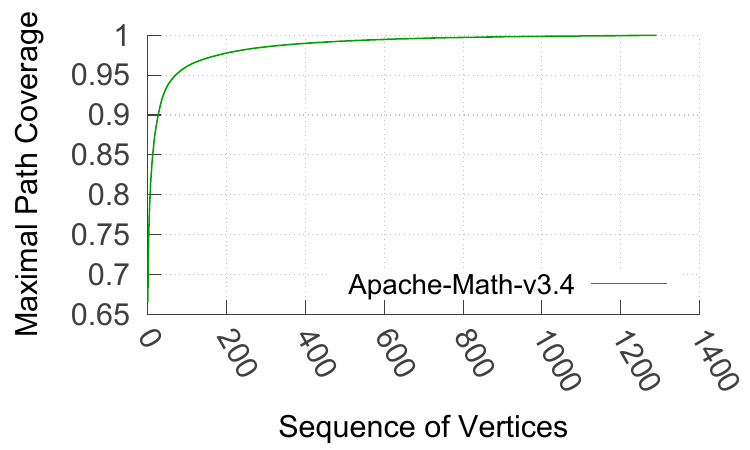}
    }
        \vspace{3mm}
  \subfigure[Rat Metabolic]
    {
        \includegraphics[scale=0.31]{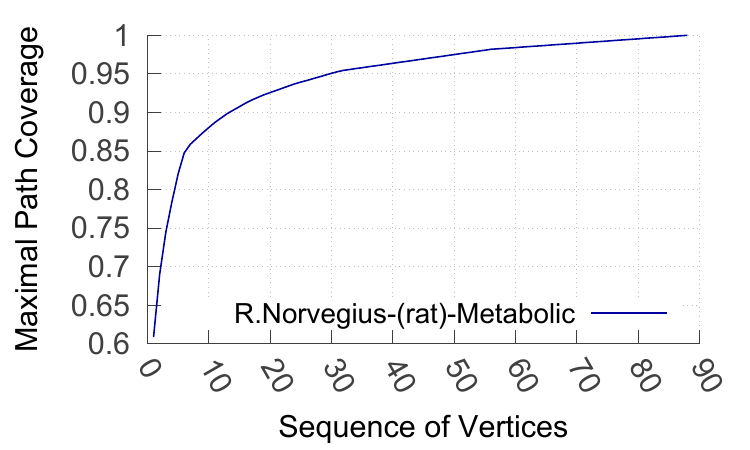}
    }
        \subfigure[Monkey Metabolic]
    {
        \includegraphics[scale=0.31]{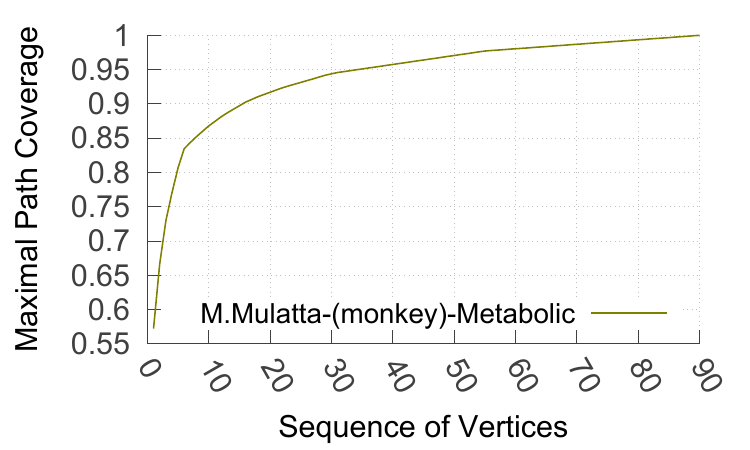}
    }
        \vspace{3mm}
        \subfigure[Abortion Cases]
    {
        \includegraphics[scale=0.31]{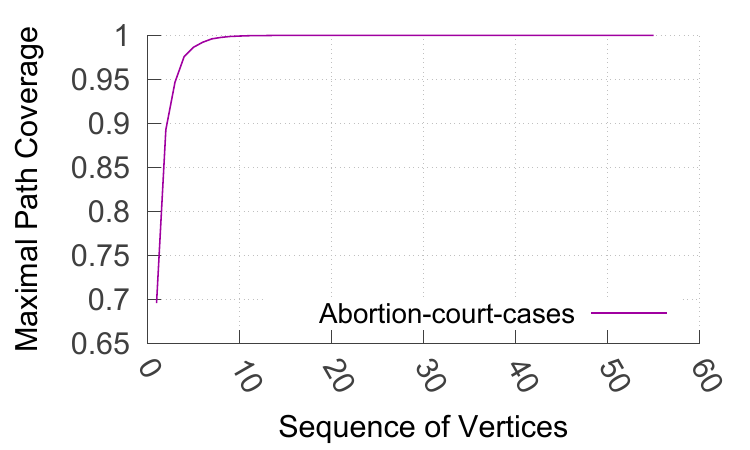}
    }
        \subfigure[Pension Cases]
    {
        \includegraphics[scale=0.31]{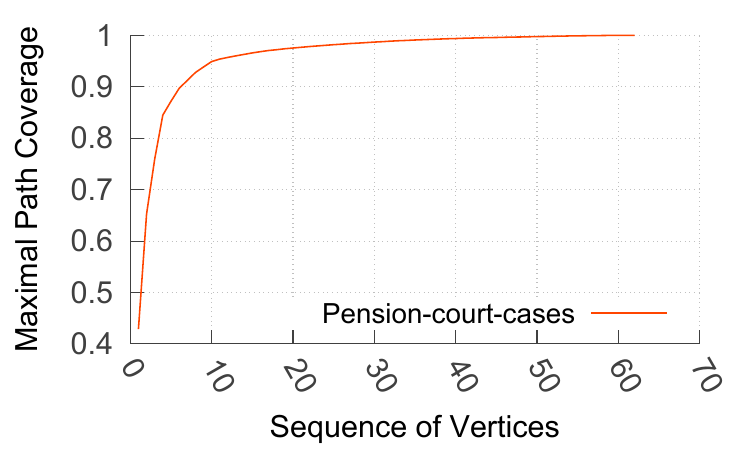}
    }

        \caption{The maximum path coverage $\hat{\delta}_k$ as a function of $k$ for the
six dependency networks.} 
    \label{fig:path-coverage-knee}
\end{figure}

Figure~\ref{fig:path-coverage-knee} shows the maximum path coverage $\hat{\delta}_k$ that
results from solving the {\bf C$^3$MC} problem iteratively, for increasing values of $k$, until $\hat{\delta}_k$
approaches 100\%.
Note that all six curves are strongly concave and that almost all ST-paths are covered with a very small
number of vertices relative to the size of each network.

\begin{figure}
    \centering
        \includegraphics[scale=0.6]{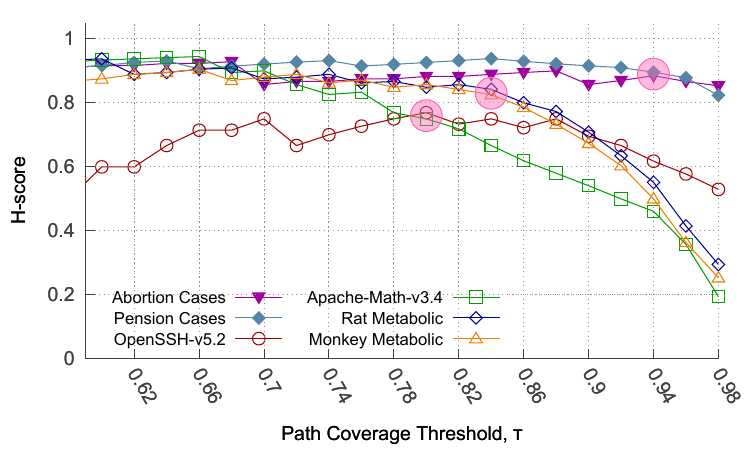}
    \caption{Effect of $\tau$ on H-score. The value of $\tau$ that we use in the rest of the analysis for 
each network is shown with a magnified symbol.}
    \label{fig:tau-effect-real}
\end{figure}

Figure~\ref{fig:tau-effect-real} examines the effect of the path coverage threshold $\tau$
on the resulting H-score of each network. 
As expected, if we require that the core covers a higher fraction of ST-paths
the core will need to be larger.
The two citation networks strongly exhibit the hourglass effect,
as their H-score remains close to 0.9 even when the core covers 90-95\% of all ST-paths.
The two metabolic networks can be also viewed as hourglass networks, 
with an H-score of about 0.85, but only as long as the core covers less than 80-85\% of all ST-paths.
Their core would need to be significantly larger to cover the remaining paths.
The two call-graphs are structured differently and they exhibit a weaker hourglass effect:
OpenSSH's H-score varies erratically between 0.6 to 0.8 depending on $\tau$, 
while the Apache Math library's H-score quickly drops below 0.8 when the core needs
to cover more than 80\% of all ST-paths. 

Based on Figure~\ref{fig:tau-effect-real}, in the rest of the analysis we set $\tau$ 
at the largest value before the H-score shows a significant drop.
After selecting the same value for each network type, we set $\tau$ as follows:
call-graphs $\tau$=80\%, metabolic networks $\tau$=85\%, and citation networks $\tau$=95\%. 

\begin{table}
\scriptsize
\centering
\begin{adjustbox}{width=1\textwidth}
\begin{tabular*}{\textwidth}{l @{\extracolsep{\fill}} cc@{\hskip 0.10in}cc@{\hskip 0.10in}cc}
                \multirow{5}{*}{Core Properties} & \multicolumn{6}{c}{Networks} \\
                \cline{2-7}
                & \multicolumn{2}{c}{Software Call-graphs} & \multicolumn{2}{c}{Metabolic Nets} & \multicolumn{2}{c}{SCotUS
 Citation Nets} \\
                & \textit{OpenSSH} & \textit{Apache Math} & \textit{Rat} & \textit{Monkey} & \textit{Abortion} & \textit{Pension} \\
                & \textit{v-5.2} & \textit{v-3.4} & & & \textit{Cases} & \textit{Cases} \\
                \hline \hline
							Path coverage threshold $\tau$ & 0.8   & 0.8   & 0.85  & 0.85  & 0.95  & 0.95 \\
\hline
                            Core size $C$    & 3     & 9     & 7     & 8     & 4     & 11 \\
                                    $C/V$    & 0.002 & 0.001 & 0.01  & 0.02  & 0.002 & 0.008 \\
                                  H-score    & 0.77  & 0.75  & 0.82  & 0.81  & 0.86  & 0.89 \\
                 Number of distinct cores    & 1     & 1     & 1     & 1     & 1     & 1 \\
                             SCCs in core    & 0     & 1     & 3     & 3     & 0     & 0 \\
                    Number of PES in core    & 0     & 2     & 1     & 2     & 0     & 0 \\
                    Core vertex coverage     & 0.35  & 0.21  & 0.53  & 0.57  & 0.82  & 0.48 \\
        	   Average core location     & 0.50  & 0.12  & 0.45  & 0.44  & 0.74  & 0.24 \\
                \hline
\end{tabular*}
\end{adjustbox}
\caption{Properties of the identified core for each dependency network.}
\label{tab:waist-summary}
\end{table}

Table~\ref{tab:waist-summary} summarizes the key properties of the core of each dependency network. 
The size of the core ${\bf C}$ varies from 0.1\% to 1\% of the network size $V$.
In all six networks we identified only one core (no ties); some vertices in the core of 
the metabolic networks and of the Apache Math network are super-vertices. 
For the selected values of $\tau$, the H-score is higher than 0.75 in all networks.

Even though the core of each network is quite small, relative to the total number of vertices,
none of these networks can be described as an ``ideal hourglass''.
This is shown both in terms of the H-score in Figure~\ref{fig:tau-effect-real} and by the core vertex coverage:  
there is a significant fraction of vertices (about 20-80\%, depending
on the network) in ST-paths that bypass the core (``tendril paths''). 
The fraction $1-\tau$ of ST-paths that bypass the core traverse at least two vertices each (a source
and a target). 
When these tendrils traverse several intermediate vertices however, the core
vertex coverage can be significantly lower than $2 \times (1-\tau)$. 
As shown in the modeling results of the next section 
(see Figure~\ref{fig:model-alpha-effect-2}-a), such low values of the core vertex coverage
can be expected when 
each vertex has a bias to depend on vertices of similar complexity with itself 
(rather than to depend directly on sources or low complexity vertices)
but where that bias is not strong enough to generate an ``ideal hourglass'' in which a
small set of intermediate vertices is traversed by all ST-paths.



\begin{figure}
        \centering
    {
        \includegraphics[scale=0.425]{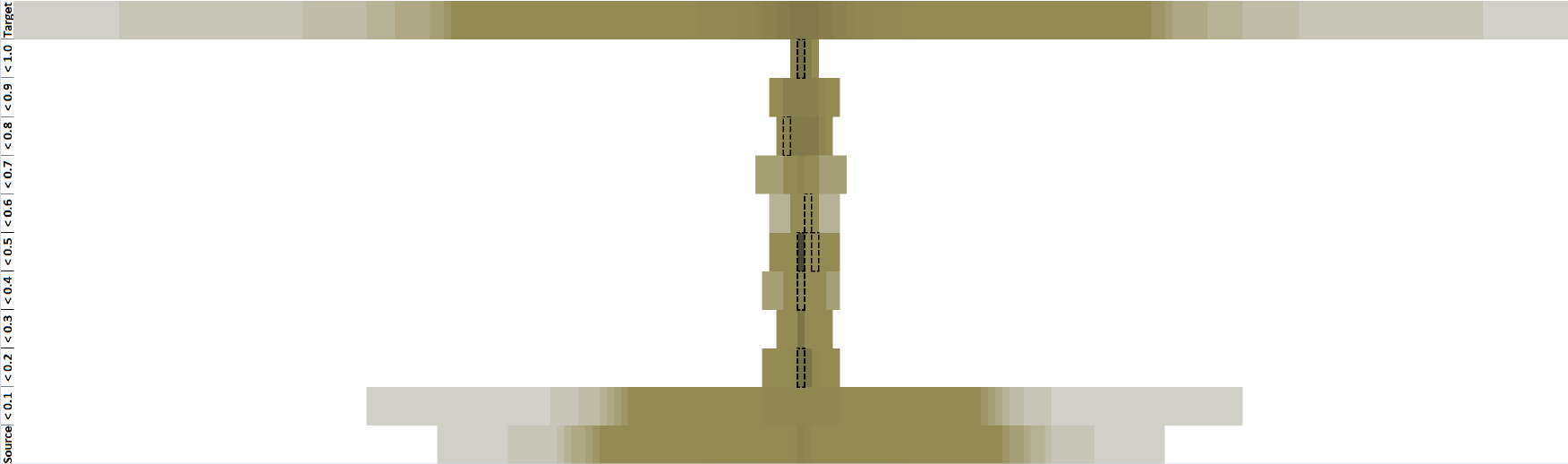}
    }
        \caption{A visualization of the Rat metabolic network that 
places vertices in the vertical direction based on their location metric.
Specifically, we discretize the location metric in 12 bins (the lowest bin for sources, the highest for targets,
and the 10 intermediate bins for intermediate vertices with each bin accounting for $1/10$ of the $0-1$ range).
The path centrality of each vertex is represented by its color (darker for higher path centrality).
Vertices with higher centrality are placed closer to the vertical midline.
The core nodes are represented by dotted rectangles.}
  \label{fig:core-locations-heatmap}
\end{figure}

Figure~\ref{fig:core-locations-heatmap} is a visualization of the Rat metabolic dependency network that 
places vertices in the vertical direction based on their location metric (see caption for more details
about this visualization). 
Note that the highest path centrality vertices tend to be at intermediate locations -- but some of the sources
and targets in this network also have high path centrality. 
Also, about half of the core vertices are 
located close to the center of the network (location=0.5), while the rest are closer to sources or targets.
The path centrality and the weight of each core vertex are shown in Table~\ref{tab:core-nodes-rat}.
The location of the core vertices varies significantly across different networks.
Similar visualizations for the other five networks are given in Figure~\ref{fig:heatmap-othernets}.

\subsection{Which are the vertices at the waist of the hourglass?}
The complete list of core vertices for each dependency network, together with
a short description, the path centrality and the weight of  core vertex, 
are given in the {\em Appendix}.
Here we comment on the qualitative properties of the waist vertices for each network.

The three vertices at the core of the OpenSSH call-graph are shown in
Table~\ref{tab:core-nodes-openssh}.
They are functions to send and receive network packets, and to execute 
Unix shell commands. 
This is not surprising given that OpenSSH is a communication-oriented utility that
can be used as a secure remote terminal, among other applications.

The Apache Math library has a core with nine vertices, listed in Table~\ref{tab:core-nodes-java}.
These methods cover floating point arithmetic operations, matrix decomposition, 
vector computations, and the ``constructors'' of some classes related 
to mathematical and geometric objects.  

The vertices at the waist of the two metabolic networks are shown in
Tables~\ref{tab:core-nodes-rat} and \ref{tab:core-nodes-monkey}.
In biochemistry,
the following twelve {\em precursors} are often considered as the most important metabolites, providing an
interface  between the different catabolic pathways with the various biosynthesis pathways:
Glucose-6-Phosphate, Fructose-6-Phosphate, Glycerone Phosphate, Glyceraldehyde 3-Phosphate, Phosphenol Pyruvate,
Pyruvate, Ribose-5-Phosphate, Erythrose-4-Phosphate, Acetyl-CoA, a-ketogluterate, Oxalocetate, and Succinyl-CoA
\cite{smolke2009metabolic, cellbiology, tanaka2005highly}; it is not clear however if these precursors are
equally important for every species or if the previous list should include additional metabolites.
In the case of Rat metabolic network, the identified waist includes eight of the previous precursors,
plus few more key compounds for the synthesis of enzymes, lipids, fatty acids, etc.
In the case of the Monkey metabolic network, the waist includes seven precursors.
Several waist vertices are the same with those in the Rat 
(or similar, in the case of SCCs or PES).

The vertices at the waist of the two citation networks are shown in
Tables~\ref{tab:core-nodes-abortion} and \ref{tab:core-nodes-pension}.
The Cornell Legal Information Institute (CLII) lists several {\em landmark} SCotUS cases for every major legal matter
in the US \cite{cornell-law}.
This classification of cases as landmarks is based on input from legal experts.
All court cases that appear in the waist of the Abortion network are also listed as landmarks by CLII.
In the Pension network, five out of the seven waist vertices are also listed as landmarks by CLII.

\section{A model of dependency network formation} \label{sec:model}
What determines whether a dependency network will exhibit the hourglass property or not? 
Let us think about this question in the context of Lego-like toys, in which 
a vertex $v$ corresponds to a Lego module and its incoming edges show which simpler Lego modules
are required to put $v$ together.
The sources correspond to the given elementary building blocks and the targets correspond
to the final objects we want to construct.  
One extreme approach is to create every object only from the elementary blocks, 
without reusing any intermediate modules that have been previously constructed.
Another approach is to reuse as much as possible intermediate modules, 
expecting that this will require less work. 
In practice, of course, the design approach is always somewhere in the middle, 
with more complex intermediate modules constructed from simpler intermediate modules as well as 
elementary blocks. 

To understand the implications of this \enquote{preference for reuse},
we present here a simple, probabilistic model for the gradual formation of a dependency network.
The model focuses on how each new vertex selects its incoming edges among the set of vertices
that have been previously constructed. 
Through a single {\em reuse parameter} $\alpha$, the model generates dependency networks in which every new vertex 
depends on either mostly sources (leading to flat, non-hourglass networks) or on the more recently constructed 
intermediate vertices (resulting in hourglass networks), or anything in between.

We refer to the following model as {\em Reuse-Preference} or {\em RP-model}. 
There are $V$ vertices that consist of $S$ sources, $M$ intermediates and $T$ targets. 
The vertices arrive in the network, or they are created, sequentially or in batches, as follows.
First, all sources are created at the same time; they represent the elementary modules of the
underlying system.
Then, the intermediate vertices are created sequentially (the case of batch arrivals is considered in
Section~\ref{sec:model-fit-real}). 
Suppose that $v$ is the $m$'th intermediate
vertex that has arrived in the network, with $1\leq m \leq M$. We assign vertex $v$ to {\em rank-0},
and the previously created $m-1$ intermediate vertices to {\em rank-1} through {\em rank-(m-1)} 
(in order of arrival -- the oldest intermediate vertex always has {\em rank-(m-1)}). 
The $S$ sources are randomly given ranks $m$ through $m+(S-1)$.
Note that the ranking changes every time a new vertex is added.
The $T$ targets are created in a batch at the end of the network formation process, 
and they are given the same rank ({\em rank-0}). 

Suppose that we are given the in-degree $d_{in}(v)$ of $v$.  
The origin of every incoming edge to $v$ is determined as follows. 
When the $m$'th intermediate vertex $v$ is created, we select the vertices it will depend on
probabilistically. In the following, we use the Zipf distribution (but other statistical models 
could also be used). Specifically, the probability that $v$ will have an 
incoming edge from a vertex $u$ at {\em rank-$r$} is given by:
\begin{equation} \label{eq:zipf}
\mbox{Prob}[(u,v)\in {\bf E}] = \frac {r^{-\alpha}} {\sum_{i=1}^{S+m-1} i^{-\alpha}}, \quad 1 \leq r \leq S+m-1
\end{equation} 
The incoming edges to the $T$ target vertices are determined in the same way; note that a target will never
by connected to another target because all targets are added in the same batch, having {\em rank-0}. 
Additionally, we artificially exclude the possibility of multi-edges.

When $\alpha=0$ the newly created vertex $v$ selects dependencies uniformly across all earlier vertices. 
As $\alpha$ increases above zero, 
$v$ has a preference for more recently constructed vertices, increasing the level of reuse 
in the dependency network.
On the other hand, as $\alpha$ decreases below zero, 
$v$ has a preference for older vertices, i.e., closer to the sources, decreasing the level of reuse.  


\textcolor{blue}{
For large values of $\alpha$, it is possible that many sources will not be chosen by any vertex higher in the hierarchy. To ensure that there are no disconnected sources (i.e., elementary blocks that are not utilized by any other module), we add an edge from every source to the first intermediate vertex, say $v$. So, instead of its originally assigned in-degree $d_{in}(v)$, vertex $v$ now has $S$ incoming edges. These extra edges however can inflate the path centrality of $v$ and of any vertices that depend on $v$. To maintain the path centrality of $v$ relative to the rest of the intermediate and target vertices, we need to increase the weight of the edges from sources to other vertices by a factor $S/d_{in}(v)$. To avoid fractional weights,  the weight of the extra edges from sources to $v$ is set to 1, the weight of the original edges from sources to other vertices is set to $S/d_{in}(v)$, and $d_{in}(v)$ is sampled so that the previous ratio is an integer. The rest of the edges have a weight of 1.
}

\begin{figure}
	\centering
	\subfigure[$\alpha$=-1, H-score=0.0]
    {
        \includegraphics[scale=0.3]{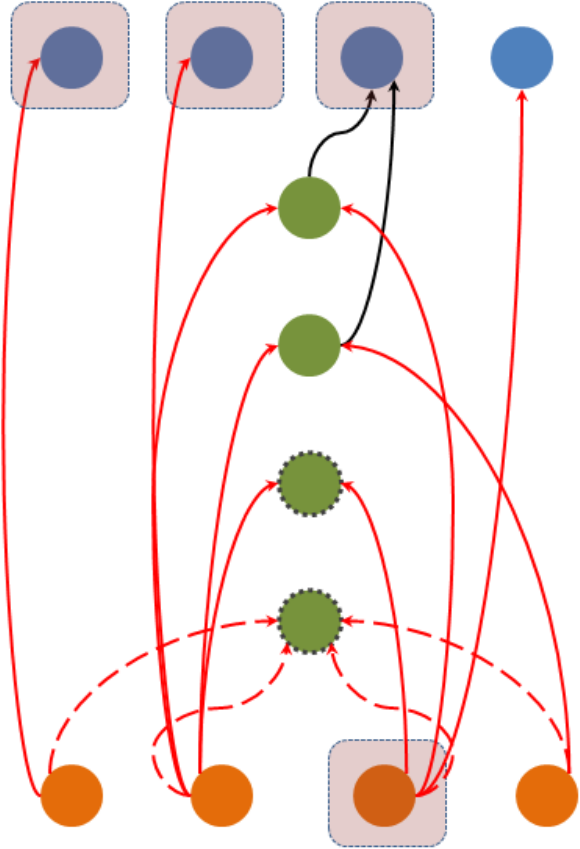}
    }
	\hspace{10.0mm}
	\subfigure[$\alpha$=0, H-score=0.25]
    {
        \includegraphics[scale=0.3]{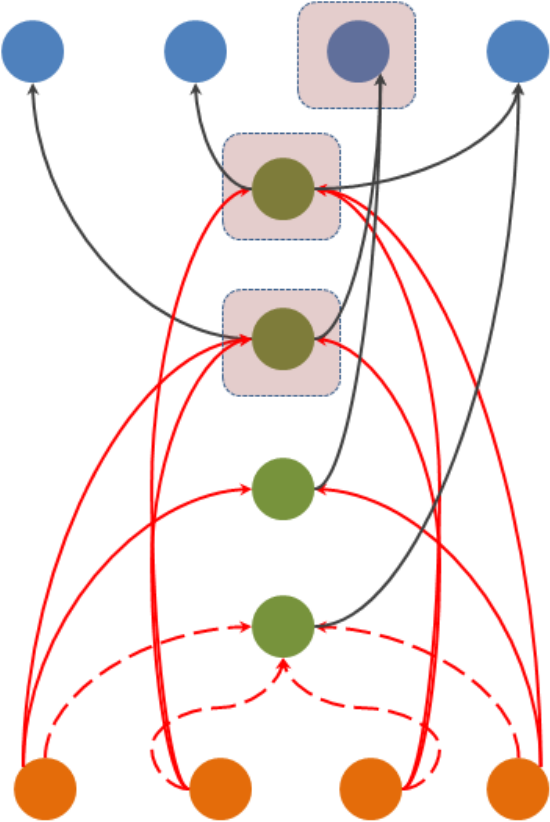}
    }
	\hspace{10.0mm}
    \subfigure[$\alpha$=+1, H-score=0.5]
    {
        \includegraphics[scale=0.3]{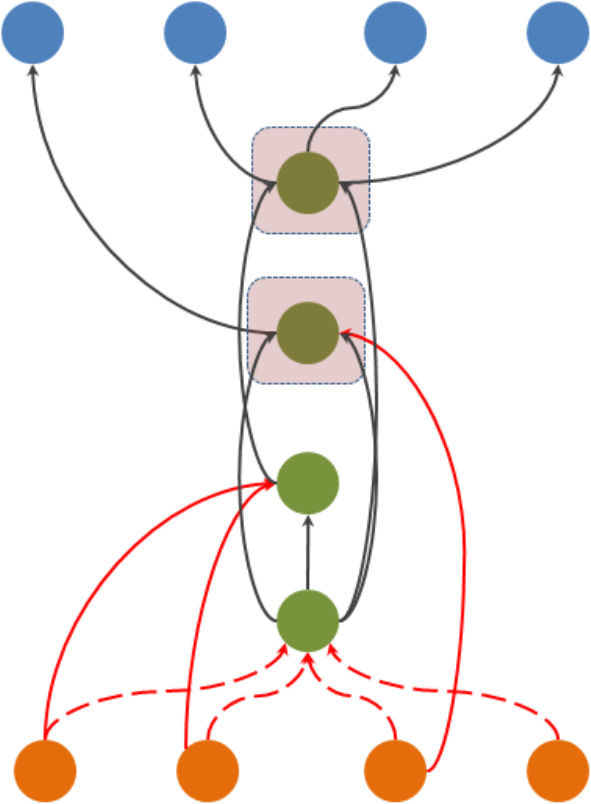}
    }
	\caption{\textcolor{blue}{Three dependency networks generated by the RP-model for different values of $\alpha$ ($V$=12, 
$S$=$T$=$M$=4, $d_{in}$=$1+\mbox{Poisson}(1)$, $\alpha$=$\{-1,0,1\}$, and $\tau$=0.90. 
The sources are shown in orange, the targets in blue, and the intermediates in green. 
Vertices that do not belong to any ST-path are shown as dotted. The edges from all sources to the first intermediate vertex $v$, shown with red dashed arrows, have unit weights. The weight of edges from sources to other intermediate and target vertices, shown with red-solid arrows, is increased to $S/d_{in}(v)$. The remaining edges, shown with solid black arrows, have unit weights. The core vertices 
for each network are shown in boxes.}} 
  \label{fig:synthetic-alpha}
\end{figure}

Figure~\ref{fig:synthetic-alpha} shows three small dependency networks constructed using the RP-model
for three different values of $\alpha$. When $\alpha=-1$, almost all ST-paths are directly connecting 
sources to targets (little reuse of intermediate vertices), 
and most intermediate vertices are not used in the construction of any target (shown as dotted).
The core consists of a combination of sources and targets, and it is relatively large (in this example, equal
to the number of sources or targets).
On the other hand, when $\alpha=+1$, 
the preference to connect to higher complexity vertices leads to longer dependency paths.
A small number of intermediate vertices are traversed by a large fraction of ST-paths, just based on chance, 
and so those vertices end up with much higher path centrality than most other vertices.
The core of such dependency networks is then small, relative to the number of sources or targets, and 
those networks have high H-score.

In the following, we illustrate the behavior of the RP-model with computational experiments.
All networks have $V$=1000 vertices but we vary the proportion of sources, targets and intermediate vertices.
The path coverage threshold $\tau$ is set to 90\%, unless stated otherwise.  
The in-degree of each vertex is either constant (denoted as ``$d_{in}$=Const(x)'') 
or set to $1+\mbox{Poisson}(x)$ where $x$ is the mean of a Poisson distribution
(denoted as ``$d_{in}$=1+P(x)''). 
All results are based on 100 simulation runs, and they are reported with 95\% confidence intervals.

\begin{figure}
	\centering
	\subfigure[Core size]
    {
				\includegraphics[scale=0.48]{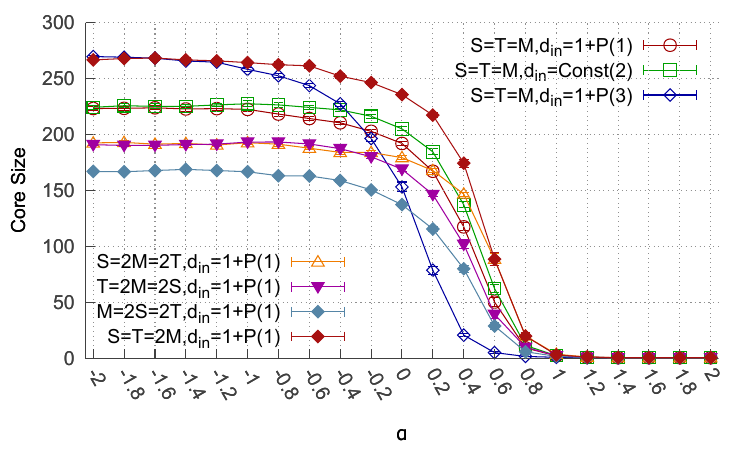}
    }
    \subfigure[H-score]
    {
				\includegraphics[scale=0.48]{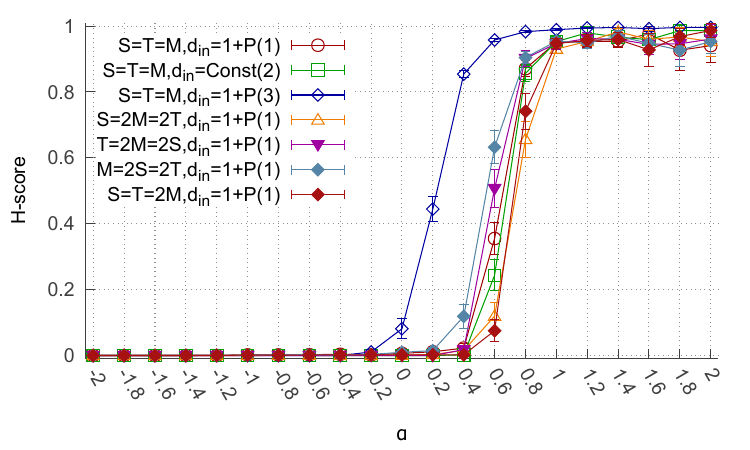}
    }
	\caption{Effect of $\alpha$ on the core size and H-score metric.}
    \label{fig:model-alpha-effect-1}
\end{figure}

\begin{figure}
	\centering
	\subfigure[Core vertex coverage]
    {
        \includegraphics[scale=0.48]{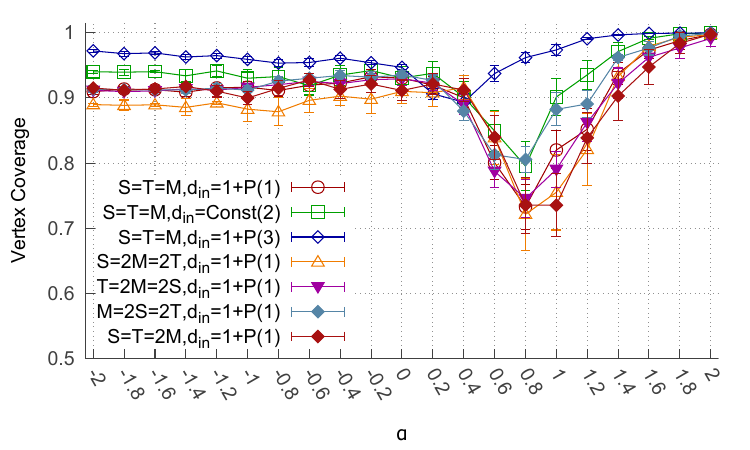}
    }
    \subfigure[Average core location]
    {
        \includegraphics[scale=0.48]{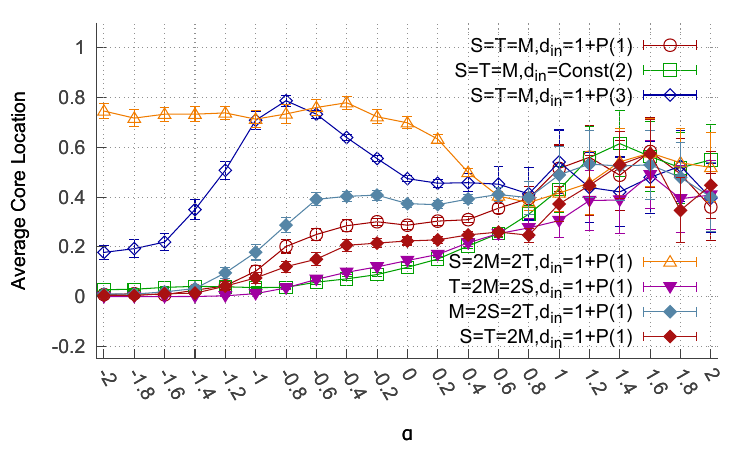}
    }
	\caption{Effect of $\alpha$ on the core vertex coverage and average core location metrics.}
    \label{fig:model-alpha-effect-2}
\end{figure}

Figures~\ref{fig:model-alpha-effect-1} and \ref{fig:model-alpha-effect-2} 
show the effect of the reuse parameter $\alpha$ on the core size $C(\tau)$, 
the H-score $H(\tau)$, the core vertex coverage $U_{\bf C}$, and the average core location $L_{\bf C}$. 
Each graph shows results for seven sets of network parameters, varying the proportion
of sources, targets, intermediates, and the in-degree values and distribution.
For example, the label $S=2M=2T$ means that $S$=500 and $M$=$T$=250 (so that $V$=1000). 

Let us first focus on negative values of $\alpha$:\\
~\\
a) As $\alpha$ decreases below zero, it becomes more likely that targets connect directly to sources
(see the ``direct'' network of Figure~\ref{fig:dependency-dag-shapes}-b or Figure~\ref{fig:synthetic-alpha}-a).
Most intermediate vertices are not included in any ST-path, their path centrality
is close to zero, and so they are not included in the core.  
Instead, the core consists of mostly a combination of sources and targets. 
To cover the large fraction $\tau$ (90\%) of these direct ST-paths however, 
the core needs to include many vertices. 
For instance, in the scenario $M$=$2S$=$2T$ the core has about 160 vertices, while $\min\{S,T\}$=250.
The higher the average in-degree is, the larger the core needs to be (to cover the increased number of ST-paths). \\
~\\
b) The corresponding flat dependency network is similar to the original network in terms
of how sources and targets are directly connected, and so it has approximately the same core size;
this is why the H-score is close to zero. \\
~\\
c) The core vertex coverage is close to one for the following reason: 
if all ST-paths are direct connections between sources and targets
and the core covers a fraction $\tau$ of these paths, 
the core vertex coverage will be at least $1-2(1-\tau)$ because 
every non-covered ST-path contributes at most two non-covered vertices.\\
~\\
d) The location of the core varies significantly with the network parameters because 
the core consists of mostly sources and targets.
So, if the core consists mostly of sources (as in the $T$=$2M$=$2S$ scenario) the core location moves closer to zero,
while if the core includes mostly targets (as in the $S$=$2M$=$2T$ scenario) the core location moves closer to one. 

Let us now focus on positive values of $\alpha$:\\
~\\
a) As $\alpha$ increases above zero, each target or intermediate vertex prefers to connect to vertices
that are close to it in the given hierarchy (see Figure~\ref{fig:synthetic-alpha}-c). 
So, the ST-paths become longer and some intermediate vertices get to be traversed by a larger
fraction of ST-paths (just based on chance). 
Vertices with high path centrality tend to form the core of the dependency network,
and their number gradually drops as $\alpha$ increases. \\
~\\
b) The core of the flat network, on the other hand, is much larger, as in the case
of negative $\alpha$, and so the corresponding H-score approaches its maximum value (one) as $\alpha$ increases.\\ 
 The transition point, from $H(\tau)\approx 0$ to $H(\tau)\approx 1$, shifts towards lower values of $\alpha$
as the density of the network increases (see scenario $d_{in}$=1+P(3)) because the likelihood 
that few intermediate vertices will acquire much higher path centrality increases.\\
~\\
c) The core vertex coverage curves follow an interesting pattern: 
as $\alpha$ increases from negative values to positive values, $U_{\bf C}$ first decreases and then increases.  
During the transition from a flat network ($H(\tau)\approx 0$) to an hourglass-like network ($H(\tau)\approx 1$), 
it is common for ST-paths to traverse one or more intermediate vertices that are not traversed by 
many other ST-paths (see the ``decoupled'' network of Figure~\ref{fig:dependency-dag-shapes}-c). 
So, in that transition range, the fraction $1-\tau$ of ST-paths that are not covered by the core 
account for more than $2(1-\tau)$ non-covered vertices (because they include one or more intermediate vertices).
As $\alpha$ further increases, the core is traversed by an increasing fraction of ST-paths, eventually covering   
almost all ST-paths, and so also covering almost all vertices that appear in ST-paths. \\
~\\
d) The location of the hourglass waist is gradually converging towards the middle of 
the dependency network, i.e., $L_{\bf C} \approx 0.5$. We should note that the location of a PES is, 
by definition, equal to the median location of the vertices in that set. So, one reason that the location
of the waist converges to 0.5 as $\alpha$ increases is that the waist in that regime often includes a large PES
with many intermediate vertices that have locations between 0 and 1.

\begin{figure}
	\centering
	\subfigure[$d_{in}=1+\mbox{Poisson}(2)$]
    {
        \includegraphics[scale=0.48]{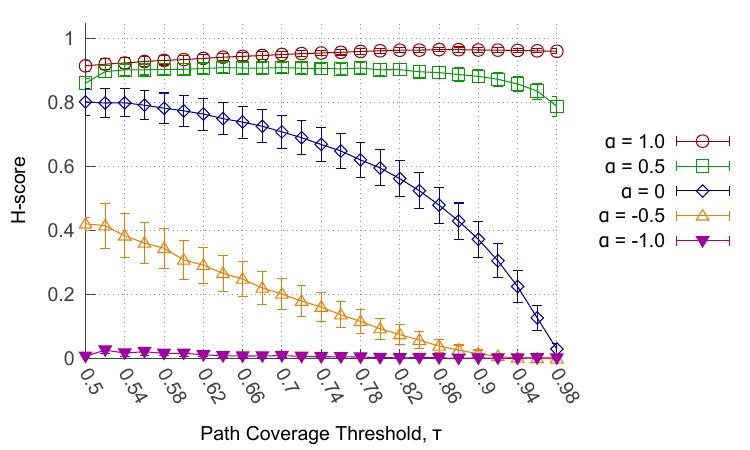}
    }
	\subfigure[$d_{in}=1+\mbox{Poisson}(3)$]
    {
        \includegraphics[scale=0.48]{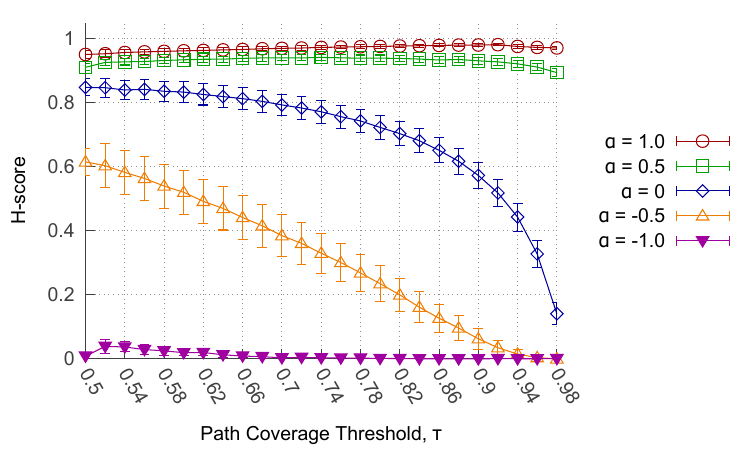}
    }
	\caption{Effect of path coverage threshold $\tau$ on H-score, for different values of $\alpha$.
Network parameters: $S$=$T$=200 and $M$=600 ($V$=1000).}
  \label{fig:tau-effect-model}
\end{figure}

Finally, Figure~\ref{fig:tau-effect-model} shows the effect of the path coverage threshold $\tau$ on the H-score
for few different values of $\alpha$.
When the reuse parameter $\alpha$ is close to one (or higher), 
the H-score is almost one, largely independent of $\tau$,
meaning that the hourglass property is robustly established.\footnote{The slight increase of the H-score with
$\tau$, when $\alpha$=1, is because the core size of the flat network increases faster than the core
size of the original network, as $\tau$ increases.}
When $\alpha$ is negative or even close to zero, on the other hand, the H-score is typically less than 50\% 
and so those networks clearly do not have the hourglass property, independent of the selection of $\tau$. 
For intermediate values of $\alpha$, the H-score depends on the selection of $\tau$ and on other network 
parameters, such as the average in-degree.

\subsection{Fitting the RP-model to a given dependency network} \label{sec:model-fit-real}

We now describe how to parameterize the RP-model so that it produces random networks $G$ that 
have approximately the same H-score with a given dependency network $G'$.
We also compare these synthetic networks with $G'$ in terms of the path
centrality distribution, the out-degree distribution, and some more network metrics that are relevant to 
dependency networks. 

Given $G'$, we can easily identify its set of sources and targets.
A synthetic network $G$ will have the same set of sources, targets, and intermediate vertices.
Since it may not be possible to identify a global ordering between the intermediate vertices, 
we place the vertices of $G$ in layers based on the topological sorting of $G'$, as follows.
First, all sources of $G'$ are placed at layer-0 of $G$. 
Then, recursively, we place at layer $i$ of $G$ those intermediate vertices of $G'$  
that depend on at least one vertex of layer $i-1$ (for $i>0$).
Finally, the targets of $G'$ are placed at the top layer of $G$ (independent of the layer of their incoming edges). 
This layered representation of $G$ gives a partial ordering relation 
between vertices: the vertices of layer $i$ are supposed to arrive (or to be created) as a batch,
and they do not depend on each other. 

The in-degree $d_{in}(v)$ of each non-source vertex $v$ in $G$ is the same with $G'$.  
To generate the specific inputs of $v$, we identify the set of ancestors $A(v)$ of $v$ in $G'$ 
-- $v$ depends directly or indirectly on these vertices.
When a vertex $v$ is created at layer $i$, it can receive incoming edges only from vertices in $A(v)$.
The selection of inputs of $v$ among the vertices in $A(v)$ is performed probabilistically based on 
Equation~\ref{eq:zipf}. 
The only difference with the original RP-model is that vertices of $A(v)$ that belong to the same
layer have the same rank, and so the same probability of being connected to $v$.\footnote{In the original RP-model, vertices are created sequentially
and so each layer (other than the boundary layers of sources and targets) includes only one vertex.}  

To parameterize the RP-model, we estimate the value of the reuse preference exponent $\alpha$ so
that the synthetic networks $G$ have an H-score that is approximately the same with that of $G'$. 
To do so, we generate 100 synthetic networks $G$ for each value of $\alpha$ and compute the average H-score
of that sample -- the optimal value of $\alpha$ is the value that gives the minimum difference from the 
H-score of $G'$. 

\begin{table}
\scriptsize
\centering
\begin{adjustbox}{width=1\textwidth}
\begin{tabular*}{\textwidth}{l @{\extracolsep{\fill}} cc@{\hskip 0.10in}cc@{\hskip 0.10in}cc}
                \multirow{5}{*}{\parbox{3cm}{\centering ~ }} & \multicolumn{6}{c}{Networks} \\
                \cline{2-7}
                & \multicolumn{2}{c}{Software Call-graphs} & \multicolumn{2}{c}{Metabolic Nets} & \multicolumn{2}{c}{SCotUS Citation Nets} \\
                & \textit{OpenSSH} & \textit{Apache Math} & \textit{Rat} & \textit{Monkey} & \textit{Abortion} & \textit{Pension} \\
                & \textit{v-5.2} & \textit{v-3.4} & & & \textit{Cases} & \textit{Cases} \\
                \hline \hline
                 $\alpha$ estimate       & 1.1           & 2.3            & 2.4           & 2.5           & 2.7            & 2.3 \\
																		     \midrule
	      \multirow{2}{*}{Core size}       & $3\pm1$       & $4\pm1$        & $4\pm4$       & $4\pm2$       & $5\pm0.5$      & $8\pm1$ \\
															           & (3)           & (9)            & (7)           & (8)           & (4)            & (11) \\
																		     \midrule
				\multirow{2}{*}{H-score}         & $0.69\pm0.05$ & $0.78\pm0.03$  & $0.76\pm0.07$ & $0.75\pm0.04$ & $0.78\pm0.02$ & $0.87\pm0.01$ \\
															           & (0.77)        & (0.75)         & (0.82)        & (0.81)        & (0.86)         & (0.89) \\
																		     \midrule
   \multirow{2}{*}{ST-path length}       & $8.8\pm0.8$   & $9.5\pm0.4$    & $11.3\pm3.5$  & $12.3\pm3.8$  & $14.3\pm0.5$   & $6.4\pm0.3$ \\
																	       & (10.4)        & (8.8)          & (8.3)         & (8.1)         & (14.1)         & (5.1) \\
																		     \midrule
	\multirow{2}{*}{Core vertex coverage}  & $0.28\pm0.04$ & $0.12\pm0.01$  & $0.32\pm0.07$ & $0.3\pm0.07$  & $0.85\pm0.1$   & $0.45\pm0.05$ \\
																		     & (0.35)        & (0.21)         & (0.53)        & (0.57)        & (0.82)         & (0.48) \\
																		     \midrule
  \multirow{2}{*}{Average core location} & $0.51\pm0.3$  & $0.05\pm0.02$  & $0.36\pm0.12$ & $0.3\pm0.12$  & $0.2\pm0.15$   & $0.29\pm0.18$ \\
																				 & (0.50)        & (0.12)         & (0.45)        & (0.44)        & (0.74)         & (0.24) \\
                \hline
\end{tabular*}
\end{adjustbox}
\caption{Fitting the RP-model to the six dependency networks of Section~\ref{sec:realnets}:
we show the $\alpha$ estimate, 
the average ST-path length, 
the core size (for the same value of $\tau$ as in the analysis of the original networks -- see Table~2), 
the core vertex coverage $U_{\bf C}$, 
and the average core location $L_{\bf C}$. 
The corresponding values for the original dependency networks are shown in parentheses.}
\label{tab:real-model-properties-2}
\end{table}

Table~\ref{tab:real-model-properties-2} shows the estimate of $\alpha$ for each dependency 
network of Section~\ref{sec:realnets}.  
The average H-score of those synthetic networks is within $10\%$ of the H-score of $G'$.

The RP-model generates ST-paths with a similar average
length as in the given dependency networks. The average length of the dependency paths is an important metric,
as it represents the typical number of intermediate vertices between a source and a target.
The synthetic networks are often similar with the given dependency networks in terms of the 
core vertex coverage and the average core location but there are also some significant 
deviations (the model overestimates the core vertex coverage of the call graphs and metabolic networks, and 
it does not predict correctly the location of the core of the SCOTUS abortion cases network).

\begin{figure}
	\centering
	\subfigure[Path centrality distribution]
    {
        \includegraphics[scale=0.48]{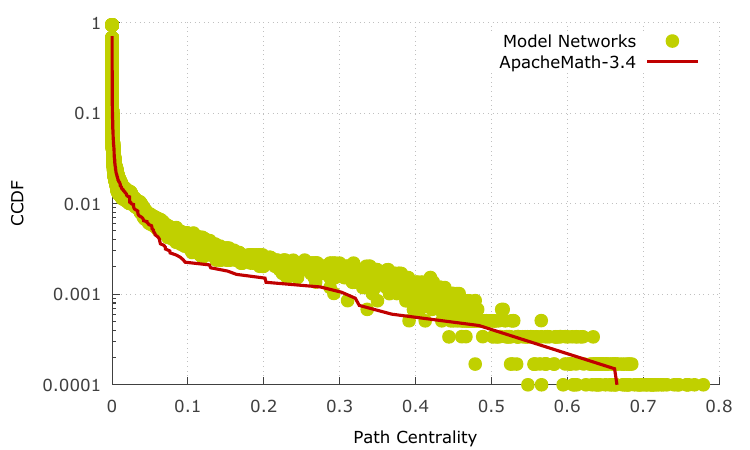}
    }
    \subfigure[Out-degree distribution]
    {
        \includegraphics[scale=0.48]{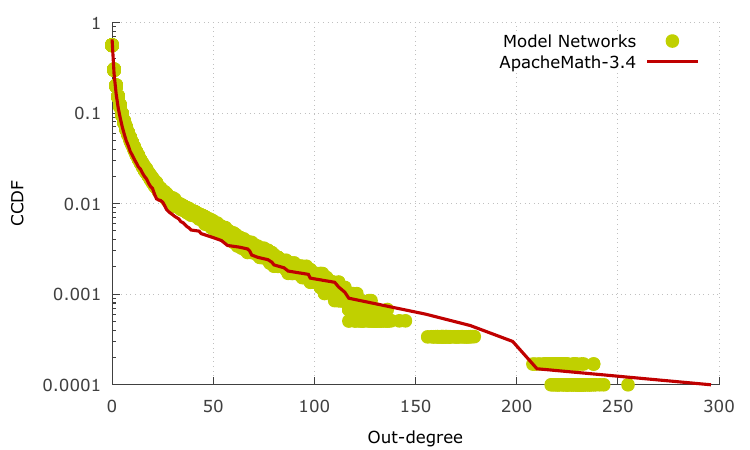}
    }
	\caption{Comparing path centrality and out degree distribution of a real network with model generated synthetic networks.}
    \label{fig:real-model-distribution-javamath}
\end{figure}

Figure~\ref{fig:real-model-distribution-javamath} shows the path centrality and the out-degree distributions
for the Apache Math call-graph and for an ensemble of 100 synthetically generated networks by the RP-model, 
as described earlier.  
Similar results for the five other dependency networks are shown at the Supplementary Material
(see Figure~\ref{fig:real-model-distribution}).
Even though there is significant variability between members of the ensemble (both distributions
are highly skewed), the model is able to generate distributions of path centrality and out-degree that 
encompass the main mass of the empirical distributions of $G'$.

\subsection{Comparison with another dependency network model} \label{sec:kleinberg}
We are not aware of any other model that can generate hourglass dependency networks.
However, there is a well-known class of models that can generate growing dependency networks based on 
variations of the ``edge-copying'' mechanism \cite{kleinberg1999web, krapivsky2005network}.
The simplest instance of the edge-copying 
model is: a new vertex $v$ depends with probability $\beta$ on a randomly
chosen vertex $u$, and with probability $1-\beta$ on a randomly chosen vertex $w$ that $u$ depends on,
i.e., $v$ copies an incoming edge of $u$ \cite{kleinberg1999web}.
If these dependencies are represented with directed edges from $u$ (or $w$) to $v$, 
the out-degree distribution follows a power-law 
with exponent $-\frac{2-\beta}{1-\beta}$ \cite{kumar2000web}.
For $\beta < 1/2$, the edge-copying model generates  scale-free networks and some some vertices are
expected to be hubs.  
An important question is: can the edge-copying model generate hourglass dependency networks,
at least for some values of $\beta$? And if so, is it that the hubs appear at the waist of the hourglass
network? 

We follow the same process as in Section~\ref{sec:model-fit-real} to fit the edge-copying model 
in a given dependency network, i.e.,  
the number of sources, targets and intermediate vertices, 
the (partial) ordering with which the vertices are created, 
and the in-degree of each vertex are as in the given dependency network.  
One special case that we need to address is: 
what if a vertex $v$ selects to copy (with probability $1-\beta$) an incoming edge of a source $u$?
Since sources do not have incoming edges, we assume that $v$ should receive an incoming edge from $u$ instead.
Also, we do not allow multi-edges. 

Figure~\ref{fig:kleinberg-hscore} shows the results of fitting the edge-copying model in the OpenSSH call-graph, Rat metabolic network and Abortion cases citation network: the y-axis shows the H-score (average and 95\% confidence interval) of 100 synthetic networks generated
for different values of the parameter $\beta$.  
Note that the H-score is close to zero throughout the range of $\beta$, meaning that the edge-copying model
is {\em not} able to generate hourglass networks. 
As $\beta$ approaches one, each new vertex depends on randomly chosen existing vertices -- which is also 
what happens in the RP-model when $\alpha=0$; we have already seen that such networks do not exhibit the 
hourglass effect.  
As $\beta$ approaches zero, the edge-copying mechanism is applied more often and this causes 
the emergence of hubs. These hubs, however, tend to be sources because the latter are created first, and so 
their number of outgoing edges increases faster than other vertices \cite{barabasi1999emergence}. 
As a result, most targets are connected directly to sources generating dependency
networks with very short ST-paths, a large fraction of disconnected intermediate vertices, and a core
that consists of almost all source vertices -- consequently, the H-score of such networks is close to zero. 

Comparing the RP-model with the edge-copying model, we note that the former is able to generate hourglass
networks, when $\alpha$ is close to one or higher, because the preference to connect to higher complexity vertices 
leads to longer dependency paths, and thus to the emergence of few intermediate vertices with much higher path 
centrality. 
 
\begin{figure}
	\centering
  \includegraphics[scale=0.5]{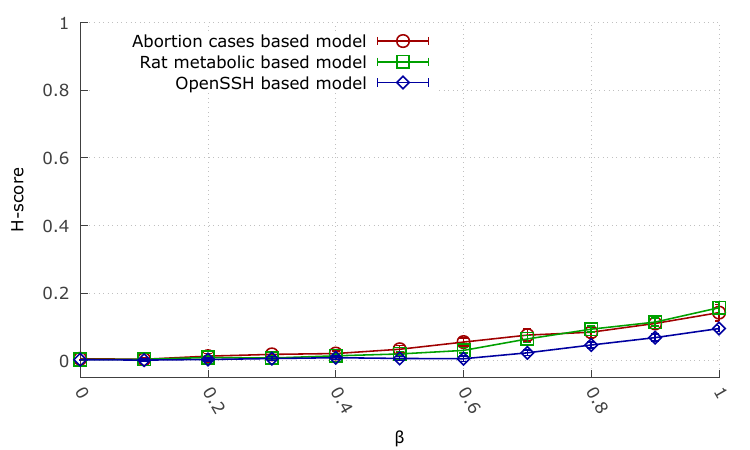}  
	\caption{H-score of synthetic networks generated by the edge-copying model. The network parameters 
(number of sources, targets, partial ordering of vertices, and in-degree of each vertex) are set based on
the three empirical dependency networks shown in the legend. } 
  \label{fig:kleinberg-hscore}
\end{figure}

\subsection{Run-time analysis of core identification algorithm} \label{sec:runtime}
We can also use the RP-model to examine the scalability of the core identification algorithm.
We created synthetic dependency networks of different sizes, for three different values of $\alpha$ (-0.5, 0, 0.5). 
The proportion of sources and targets remains constant (25\% each), while the in-degree of each non-source vertex
is $1+\mbox{Poisson}(2)$.

As discussed in Section~\ref{sec:core}, the core identification greedy algorithm has a run-time complexity
of $O(k \, E)$, where $k$ is the size of the core and $E$ is the number of network edges. 
In the dependency networks we construct, $E$ increases proportionally with the number of vertices $N$, 
i.e., $E= d_{in} \, (N-S)$, where $d_{in}$ is the average in-degree for non-source vertices and $S$ is the
number of sources.
The relation between $k$ and $N$ is not something we could derive analytically, and it certainly depends on $\alpha$
and the path coverage threshold $\tau$.  

Figure~\ref{fig:run-time} shows the run-time, the run-time per core vertex, and the core size $k$ as a 
function of $N$, for $\tau$=0.90. 
Note that $k$ increases almost linearly with $N$ for all values of $\alpha$ we consider. 
Consequently, the total run-time becomes the product of two linear functions of $N$, and so 
it {\em increases quadratically with the network size.}  
As expected, non-hourglass networks (e.g., when  $\alpha=-0.5$) have a larger core, and so 
they require more computation than hourglass networks. 

\begin{figure}
	\centering
	\subfigure[Run-time]
    {
        \includegraphics[scale=0.48]{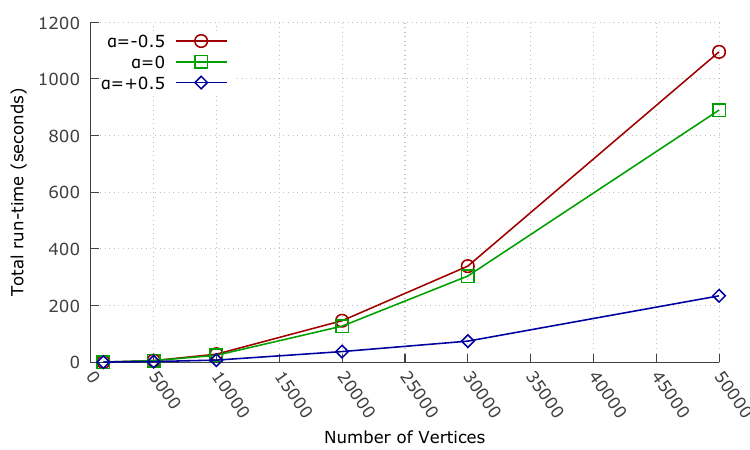}
    }
    \subfigure[Core-size and run-time per core vertex]
    {
        \includegraphics[scale=0.48]{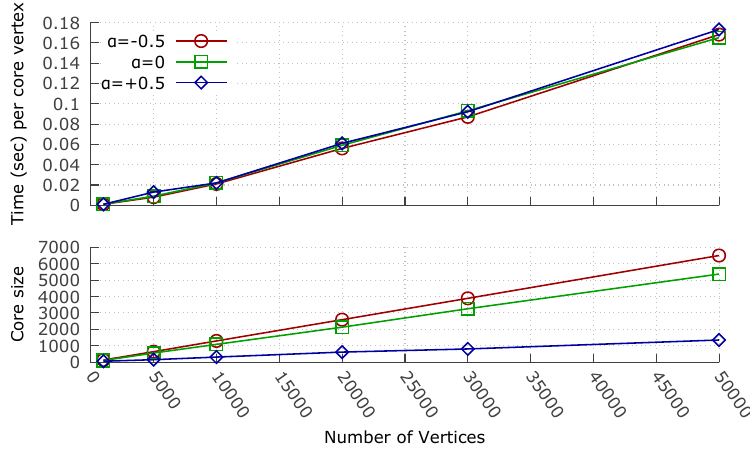}
    }
	\caption{Run-time analysis of the core identification algorithm using networks generated with the RP-model. 
The run-time increases quadratically with the network size $N$. 
The experiments were run on an Intel-2.5GHz dual-core processor with $6$GB of memory.}
  \label{fig:run-time}
\end{figure}

%

\section{Related work} \label{sec:rlwork}
The terms ``hourglass'' and ``bow-tie'' are often mentioned informally in 
the network science literature and in other disciplines -- their precise
meaning and whether the two terms are synonymous is not always clear however.

The term bow-tie, in particular, always refers to directed (but not necessarily acyclic) networks.   
It first appeared in the context of the WWW graph, after the 2000 
study of Broder et al. \cite{broder2000graph}. 
The ``knot'' of that bow-tie was described as the largest 
Strongly Connected Component (SCC) in the graph, which included about 25\% of the network's vertices. 
Similarly, the term bow-tie has been also used in the context of metabolic networks
\cite{ma2003connectivity,zhao2006hierarchical}.
Since then, several directed networks have been described as bow-ties, as long as there is a central SCC 
with incoming edges from a large input component and outgoing edges to a large output component 
\cite{newman2002email,capocci2006preferential,saito2007large,vitali2011network,easley2012networks}. 
In other words, the term  ``bow-tie network'' refers mostly to a visual representation of directed networks 
based on the previous decomposition of vertices into four sets: 
an input component, a core (the largest SCC), an output component, and any other vertices 
that are not in the previous three components (referred to as ``tendrils'' and disconnected components). 
There is no requirement that the vertices in the knot of the bow-tie account for only a small 
fraction of the network size. There is also no requirement that the vertices in the knot are 
highly central, for any definition of centrality. 

Hourglass networks, on the other hand, are typically directed and acyclic graphs, and the vertices 
at the hourglass waist need to be a small fraction 
of the total number of vertices. Further, the few vertices at the waist are present in almost all 
source-target paths in the network, and so they can be thought of as functionally
very important for the underlying system \cite{akhshabi2011evolution,csete2004bow}.

The three most relevant studies about the hourglass effect 
focused on the special case of layered and acyclic directed networks in which edges
can only exist between successive layers
\cite{akhshabi2011evolution,akhshabi2014explanatory,friedlander2015evolution}.
In those studies, the hourglass effect is defined in terms of the number of vertices at each layer,
and a network is referred to as an hourglass if the width of the intermediate layers is much smaller 
relative to the width of the input and output layers.  
The first study \cite{akhshabi2011evolution} proposed an evolutionary model (called EvoArch) for the 
emergence of the hourglass effect in computer networking protocol stacks; EvoArch captures
the creation and competition between modules that perform similar functions and it may
be also applicable in other layered technological systems.
The second study \cite{akhshabi2014explanatory} made the case that 
the topological structure of developmental regulatory networks (namely that the specificity
of regulatory interactions increases during embryogenesis) is sufficient for the emergence
of the hourglass effect in that context.  
The third study \cite{friedlander2015evolution} showed that a layered and directed network 
can evolve to a bow-tie structure if the relation between inputs and outputs can be represented
with a rank-deficient matrix, and if the mutations in the intensity (weights)
of module interactions (network edges) can be modeled as products by a random number (rather than sums). 

The previous models and analysis frameworks, however, are not applicable in more general 
dependency networks.
Even if we artificially place vertices in layers based on topological sorting (i.e., sources
are placed at the bottom layer, and each vertex is placed at the lowest possible layer so that all its 
incoming edges are from vertices of lower layers), edges can traverse more than one layer,
and targets can appear at different layers. 
Additionally, a general dependency network may include cyclic dependencies and SCCs of interdependent modules. 
So, those studies do not define the hourglass property in general hierarchically modular 
systems and they do not how show how to identify their waist.

In the context of DAGs, a relevant prior study is \cite{ishakian2012framework}. 
That work had a different focus (not related to the hourglass effect or modeling hierarchical systems) 
but it considered the same centrality 
metric (referred to as \#P centrality) that we also use, and it analyzed the computational complexity
of the problem of identifying the $k$ vertices that have, collectively as a group, the largest \#P 
centrality (referred to as the {\em C$^3$MC} problem in our work).

Another relevant study is the {\em BowTieBuilder} algorithm \cite{supper2009bowtiebuilder}. 
That work examined to what extent signal transduction pathways follow the bow-tie structure
proposing a centrality metric (``bow-tie score'') 
for each protein in the network, based on the number of sources and targets that are 
connected with paths traversing that protein.  
The knot of the bow-tie was defined as the set of proteins with maximal bow-tie score.

The ``morphospace'' of all possible hierarchical networks was investigated in 
\cite{corominas2013origins}. The three dimensions of the considered morphospace in that study
are ``treeness'', ``feedforwardness'' and ``orderability''. A large number of networks, mostly 
metabolic, neuronal and language, are shown to fall in the part of the morphospace that
corresponds to hourglass or bow-tie networks.

\section{Discussion -- significance of the hourglass effect} \label{sec:discussion}

The hourglass effect is significant for several reasons. One of them is that the modules at the waist
of the network create a ``bottleneck'' in the flow of information from sources (or inputs) to targets
(or outputs). Such bottleneck network effects have been studied in the literature under different names.
For instance, the term ``core-periphery networks'' has been broadly used in network science to
refer to various static and dynamic topological properties (e.g., rich-club effect, 
onion-like networks) that result from a dense, cohesive core that
is connected to sparsely connected peripheral vertices (but not necessarily organized in an 
acyclic input-output hierarchy) \cite{borgatti2000models,csermely2013structure,rombach2014core}. 
Bottlenecks have been also observed in gene regulatory networks \cite{bhardwaj2010analysis}, 
in protein networks \cite{yu2007importance},
in general evolutionary models \cite{jain2002large}, among many other domains. 
\revision{
The methodology we have presented in this paper for the identification
of the core and for the quantification of the hourglass effect can serve as a unified approach 
for the study of bottleneck network phenomena in a wide range of disciplines. 
}

Why do so many networks in nature and
technology exhibit the hourglass effect? Is there a single underlying explanation or are there
different mechanisms through which a hierarchical network can acquire this property?
In technological networks, the reuse of existing modules has economic benefits in terms of 
design and implementation cost, and so it may be that the hourglass property results ``by design''
\cite{yan2010comparing}.
In natural networks, on the other hand, are there similar costs that an evolutionary process 
gradually reduces or should we look for a completely different explanation?
The model of \cite{friedlander2015evolution}
captures how a realistic evolutionary process searches for the network  
that results in a desired input-output (linear) transformation.
A more recent work \cite{siyari2016lexis}
proposes an optimization-based framework, modeling sources as characters and targets as strings,
that creates the given targets through the construction and reuse of intermediate substrings.  
\revision{
The proposed RP-model offers a different, probably more general explanation for the hourglass effect: 
a dependency network with multiple sources and targets exhibits the hourglass effect 
when each vertex tends to depend on vertices of similar complexity
(instead of connecting directly to sources or vertices of much lower complexity). 
This ``preference for reuse'' tends to create deep hierarchies in which a small set of
intermediate vertices is traversed by most dependency paths. 
The RP-model is probabilistic, and so it is not possible to predict which specific intermediate 
vertices will emerge at the waist.  
In practice, we expect that the vertices at the waist will correspond to modules that are both
highly general (meaning that their function is needed, directly or indirectly, by many targets)
and highly complex (meaning that to provide that function, those modules need to utilize, directly or indirectly,
the functionality of many sources).
} 

The hourglass effect is also significant for the evolvability and robustness of hierarchically modular systems.  
Intuitively, the hourglass effect should allow a system to accommodate frequent changes in its sources
or targets (i.e., to be able to evolve as the environment changes) because the few modules at the waist 
``decouple'' the large number of sources from the large number of targets. If there is a change in the 
inputs (sources), the outputs do not need to be modified as long as the modules at the waist can 
still function properly. Similarly, if there is need for a new target, it may be much easier (or cheaper)
to construct it reusing the modules at the waist rather than directly relying on sources.
This is related to the notion of ``constraints that de-constrain'', introduced by Kirschner and Gerhart in 
the context of biological development and evolvability  \cite{kirschner1998evolvability}.  
At the same time however, the presence of these critical modules at the waist (the ``constraints'')
limit the space of all possible outputs that the system can generate (``phenotype space''), at least for a 
given maximum cost. The mechanisms through which the hourglass effect can improve evolvability but also 
limit the phenotype space is an important issue not only for natural
systems but also for evolving technological systems \cite{rexford2010future}.  

Finally, understanding the implications of the hourglass effect for the cost, robustness,
and evolvability of designed or technological systems can also have significant practical applications. 
In engineering, the primary focus is typically on optimality rather than on evolvability or robustness
(e.g., design the minimum cost electronic circuit that can perform a given logic function). 
Such system-wide cost minimizations may appear attractive at first but they typically lead to non-hierarchical
(monolithic) designs that are hard to test, evolve, or operate in the presence of failures. 
On the other hand, hierarchical design often lacks a systematic framework and the tools that would
allow the designer to automatically identify, given a set of inputs and a set of outputs, the intermediate
modules that would be most reusable.
This becomes an even harder problem when we consider that most technological systems need to evolve
as the desired functionalities (outputs) and conditions (inputs) often change over time. 
One approach, which has not been pursued so far to the extent of our knowledge, is to start the design 
process from the waist, rather than bottom-up or top-down: first design a relatively small number of 
modules of intermediate complexity that will form the waist of the dependency network. Then, construct 
these modules based on the inputs, and in parallel construct the outputs based on these modules at the waist.  
Of course the key challenge in this approach is to develop algorithms and tools that can automatically
identify those few central building blocks that will form the hourglass waist
from the system specifications.

\section*{Acknowledgment}
This research was supported by the National Science Foundation (NSF award CNS-1319549). We are also grateful to Saamer Akhshabi (Georgia Tech), Payam Siyari (Georgia Tech), Mathieu Nassif (McGill), Prof. Bistra Dilkina (Georgia Tech) and Prof. Martin Robillard (McGill) for their valuable help and input. We are also grateful to the anonymous reviewers for their thoughtful and constructive comments.

Dedicated to the memory of Saamer Akhshabi.

\bibliographystyle{plain}
\bibliography{hourglass}

\newpage
\section*{Appendix -- Supplementary material}

\clearpage
\subsection*{Notation}

\begin{table}[!ht]
\scriptsize
\centering
\captionsetup{font=footnotesize}
\begin{tabular}{p{0.20\textwidth}p{0.8\textwidth}} 
\midrule
Symbol & Description\\
\midrule
$X$ & the cardinality of a set ${\bf X}$\\ 
${\bf G_0}$ & the original directed network (may include cycles)\\
${\bf G}$ & dependency network (directed and acyclic, by construction) \\
${\bf V}$ & set of vertices\\
${\bf E}$ & set of edges\\
${\bf S}$ & set of sources\\
${\bf T}$ & set of targets\\
${\bf M}$ & set of intermediates\\
${\bf I}(v)$ & set of vertices with edges to $v$ -- inputs of $v$\\
${\bf O}(v)$ & set of vertices with edges from $v$ -- outputs of $v$\\
$d_{in}(v)$ & in-degree of $v$\\
$d_{out}(v)$ & out-degree of $v$\\
$p(s,t)$ & a path from a source $s$ to a target $t$ -- an ST-path\\
$P(v)$ & path centrality of $v$\\
$P_S(v)$ & number of paths from sources to $v$ (complexity of $v$)\\
$P_T(v)$ & number of paths from $v$ to targets (generality of $v$)\\
${\bf P}$ & set of all ST-paths\\
${\bf P_R}$ & set of ST-paths that traverse a set ${\bf R}$ of vertices\\
$\delta_{\bf R}$ & path coverage of ${\bf R} (=P_R/P)$\\
${\bf \hat{R}}_k$ & set of k vertices with maximum path coverage\\
$\hat{\delta}_k$ & path coverage of ${\bf \hat{R}}_k$\\
$\tau$ & path coverage threshold\\
${\bf C}(\tau)$ & a core (there may be more than one) for a given path coverage threshold $\tau$\\
$\delta_{{\bf C}(\tau)}$ & the path coverage of ${\bf C}(\tau)$ (may be more than $\tau$)\\
${\bf G_f}$ & the flat network that corresponds to ${\bf G}$\\
${\bf C_f}(\tau)$ & the core of the flat network ${\bf G_f}$ for the threshold $\tau$\\
$H(\tau)$ & H-score for the threshold $\tau$\\
$U_{\bf C}$ & core vertex coverage of core ${\bf C}$\\
${\bf V_{ST}}$ & set of vertices in at least one ST-path\\
$\phi_{\bf C}(v)$ & indicator variable that vertex $v \in {\bf V_{ST}}$  is reachable from, or can reach, at least one
vertex in core ${\bf C}$\\
$L(v)$ & location of vertex $v$\\
$L_{\bf C}$ & average location of core ${\bf C}$\\
$\delta_{\bf C}(v)$ & incremental path coverage when vertex $v$ is added in core ${\bf C}$ (also referred to as
the ``weight'' of core vertex $v$)\\
$\alpha$ & reuse preference exponent\\
$d_{in}$ & average in-degree\\
$\beta$ & parameter of edge-copying model\\
\bottomrule
\end{tabular}
\caption{List of symbols.}
\label{tab:list-of-symbols}
\end{table}

\clearpage
\subsection*{Submodularity of the {\bf C$^3$MC} objective function}
\begin{lemma}
The objective function of the {\bf C$^3$MC} problem is submodular, i.e.,
\begin{equation}
\delta_{{\bf X}\cup\{v\}} - \delta_{\bf X} \geq \delta_{{\bf Y}\cup\{v\}} - \delta_{\bf Y} \label{eq:submodular}
\end{equation}
for any ${\bf X},{\bf Y}$ such that ${\bf X}\subseteq {\bf Y}\subseteq {\bf V}$ and for any vertex $v$ in ${\bf V}$.
\end{lemma}
\begin{proof}
The function $\delta_{\bf R}$ is non-negative and non-decreasing.
\begin{itemize}
\item{Case 1: Consider all ST-paths that traverse $v$ but not any vertex in ${\bf Y}$. These paths
do not traverse any vertex in ${\bf X}$ either. So the increase in the coverage of ${\bf X}$ and ${\bf Y}$
will be the same when we add $v$ in both sets.}
\item{Case 2: Consider all ST-paths that traverse $v$ as well as one or more vertices of ${\bf Y}$
but not any vertex of ${\bf X}$. Such ST-paths increase the coverage of only ${\bf X}\cup{v}$.}
\item{Case 3: Consider all ST-paths that traverse $v$ as well as one or more vertices of ${\bf X}$.
These paths are already included in the coverage of both ${\bf X}$ and ${\bf Y}$, and so they will not cause any further
coverage increase by including $v$ in ${\bf X}$ and ${\bf Y}$.}
\end{itemize}
These three cases account for all ST-paths traversing $v$ and we have shown that the submodularity
condition is satisfied in all of them.
\end{proof}

\clearpage

\subsection*{Vertices at the waist of each dependency network}

\begin{table}[!ht]
\scriptsize
\centering
\captionsetup{font=footnotesize}
\begin{tabular}{p{0.2\textwidth}p{0.12\textwidth}p{0.12\textwidth}p{0.56\textwidth}} 
\midrule
		Name & $\frac{P(v)}{P}$ & $\delta_{C}(v)$ & Description\\
\midrule
		packet\_send & 0.50 & 0.50 & \textit{Wrapper function for formatting and sending TCP packet.}\\
		packet\_read\_seqnr & 0.37 & 0.20 & \textit{Function to return type of received packet.}\\
		do\_exec & 0.42 & 0.10 & \textit{Function responsible for spawning a sub-shell as part of session creation}\\ 
\bottomrule
\end{tabular}
\caption{The waist of the OpenSSH-v5.2 call-graph network.}
\label{tab:core-nodes-openssh}
\end{table}

\begin{table}[!ht]
\scriptsize
\centering
\captionsetup{font=footnotesize}
\begin{tabular}{p{0.28\textwidth}p{0.07\textwidth}p{0.07\textwidth}p{0.58\textwidth}} 
\midrule
\midrule
		Name & $\frac{P(v)}{P}$ & $\delta_{C}(v)$ & Description\\
\midrule
		SCC-1 & 0.60 & 0.60 & \textit{Methods from the decimal floating point library class.}\\
		Vector3D:init & 0.08 &  0.06 & \textit{Initializer for base class implementing vectors in a three-dimensional space.} \\
		DerivativeStructure:init & 0.10 & 0.04 & \textit{Initializer for base class that is the workhorse of differentiation library.}\\
		FastMath:abs & 0.04 & 0.03 & \textit{Faster math library's absolute value computing method.}\\
		EigenDecomposition:init & 0.10 & 0.02 & \textit{Initializer for the class handling eigen decomposition of a real matrix.}\\
	  BigFraction:init & 0.05 & 0.02 & \textit{Initializer for base class representing a rational number without any overflow.} \\
		MatrixUtils:createRealMatrix & 0.09 & 0.01 & \textit{Method to create and initialize a real-valued matrix from given data.}\\
		Line:init & 0.04 & 0.01 & \textit{Initializer for the three dimensional geometric line Java class.}\\
		IntervalsSet:iterator & 0.008 & 0.01 & \textit{Iterator for traversing a set of one dimensional geometric intervals.}\\
		
\bottomrule
\end{tabular}
\caption{The waist of the Apache-Math-v3.4 call-graph network. SCCs are listed in Table~\ref{tab:apache-math-scc}.}
\label{tab:core-nodes-java}
\end{table}

\begin{table}[!ht]
\scriptsize
\centering
\captionsetup{font=footnotesize}
\begin{tabular}{p{0.1\textwidth}p{0.9\textwidth}}
		\midrule
		SCC & Components\\
		\midrule
		SCC-1 & 
		\textit{DfpMath:splitPow, Dfp:lessThan, Dfp:align, DfpMath:split, Dfp:dotrap, Dfp:multiply, Dfp:divide, DfpMath:log, Dfp:multiplyFast, DfpMath:logInternal, Dfp:negate, Dfp:add, Dfp:remainder, Dfp:init, Dfp:power10K, Dfp:trunc, Dfp:unequal, DfpMath:splitMult, Dfp:subtract, DfpMath:exp, Dfp:floor, Dfp:newInstance, DfpField:newDfp, Dfp:toDouble, Dfp:rint, Dfp:greaterThan, Dfp:round, DfpMath:pow, DfpMath:expInternal, Dfp:intValue, Dfp:copysign} \\
\bottomrule
\end{tabular}
\caption{SCCs in the core of the Apache-Math-v3.4 call-graph network.}
\label{tab:apache-math-scc}
\end{table}

\begin{table}[!ht]
\scriptsize
\centering
\captionsetup{font=footnotesize}
\begin{tabular}{p{0.22\textwidth}p{0.09\textwidth}p{0.09\textwidth}p{0.60\textwidth}} 
\midrule
\midrule
		Name & $\frac{P(v)}{P}$ & $\delta_{C}(v)$ & Description\\
\midrule
		SCC-1 & 0.60 & 0.60 & \textit{Contains the metabolic precursors: Pyruvate, PhosphenolPyruvate, Oxalocetate.}\\
		Arachidonate & 0.08 & 0.09 & \textit{Essential for enzyme synthesis.}\\ 
		Acetyl-CoA & 0.25 & 0.05 & \textit{A metabolic precursor.}\\
		SCC-2 & 0.25 & 0.04 & \textit{Contains the metabolic precursors: Glycerone Phosphate, Ribose-5-Phosphate, Glyceraldehyde 3 Phosphate.}\\
		Phosphatidate & 0.07 & 0.04 & \textit{Essential for Lipid synthesis.}\\ 
		SCC-3 & 0.02 & 0.02 & \textit{These compounds take part in Purine metabolism.}\\
		\{GQ1b, Glycan 9-11\} & 0.01 & 0.01 & \textit{These compounds take part in Ganglioside metabolism.}\\
\bottomrule
\end{tabular}
\caption{The waist of the Rat (\textit{R. Norvegius}) metabolic network. SCCs are listed in Table~\ref{tab:rat-metabolic-scc}.}
\label{tab:core-nodes-rat}
\end{table}

\begin{table}[!ht]
\scriptsize
\centering
\captionsetup{font=footnotesize}
\begin{tabular}{p{0.1\textwidth}p{0.9\textwidth}}
		\midrule
		SCC & Components\\
		\midrule
		SCC-1 & 
		\textit{Ammonia, Pyruvate, Oxalocetate, L-Alanine, L-Aspartate, Glutathione, Glycine, L-Arginine, L-Glutamine, L-Serine, Phosphoenolpyruvate, gamma-Glutamylcysteine, L-Argininosuccinate, Mercaptopyruvate, L-Cystathionine, Casbene, Carbomoyl Phosphate, Fumarate, Citrolline, Malate, L-Cysteine, L-Ornithine} \\
		SCC-2 & \textit{Glycerone Phosphate, Ribose 5-Phosphate, Glyceraldehyde 3-Phosphate, PRPP, D-Xylulose 5-Phosphate, beta-D-Fructose-6-phosphate, Sedoheptulose 7-phosphate, beta-D-Fructose 1,6-bisphosphate} \\
		SCC-3 & \textit{AMP, GDP, DNA, Guanine, Deoxyadenosine, dATP, IMP, dGTP, XMP, Xanthine, Guanosine, dADP, dAMP, dGDP, Inosine, Adenine, Hypoxanthine} \\
\bottomrule
\end{tabular}
\caption{SCCs in the core of the Rat (\textit{R. Norvegius}) metabolic network.}
\label{tab:rat-metabolic-scc}
\end{table}

\begin{table}[!ht]
\scriptsize
\centering
\captionsetup{font=footnotesize}
\begin{tabular}{p{0.26\textwidth}p{0.09\textwidth}p{0.09\textwidth}p{0.56\textwidth}} 
\midrule
\midrule
		Name & $\frac{P(v)}{P}$ & $\delta_{C}(v)$ & Description\\
\midrule
		SCC-1 & 0.57 & 0.57 & \textit{Contains the metabolic precursors Pyruvate, PhosphenolPyruvate, Oxalocetate.}\\
		Arachidonate & 0.10 & 0.09 & \textit{Essential for enzyme synthesis.}\\ 
		Acetyl-CoA & 0.25 & 0.06 & \textit{A metabolic precursor.}\\
		Phosphatidate & 0.08 & 0.05 & \textit{Essential for lipid synthesis.}\\
		SCC-2 & 0.21 & 0.03 & \textit{Contains the metabolic precursors: Glycerone Phosphate, Ribose-5-Phosphate, Glyceraldehyde 3 Phosphate.}\\
		SCC-3 & 0.03 & 0.03 & \textit{These compounds take part in Purine metabolism.}\\
		Lc3Cer & 0.01 & 0.01 & \textit{Aids in biosynthesis of Glycolipids.}\\
		Malonyl-[acp] & 0.01 & 0.01 & \textit{A key compound for fatty acid synthesis.}\\
\bottomrule
\end{tabular}
\caption{The waist of the Monkey (\textit{M. Mulatta}) metabolic network. SCCs are listed in Table~\ref{tab:monkey-metabolic-scc}.}
\label{tab:core-nodes-monkey}
\end{table}

\begin{table}[!ht]
\scriptsize
\centering
\captionsetup{font=footnotesize}
\begin{tabular}{p{0.1\textwidth}p{0.9\textwidth}}
		\midrule
		SCC & Components\\
		\midrule
		SCC-1 & 
		\textit{Ammonia, Pyruvate, Oxalocetate, L-Alanine, L-Aspartate, Glutathione, Glycine, L-Arginine, L-Glutamine, L-Serine, Phosphoenolpyruvate, gamma-Glutamylcysteine, L-Argininosuccinate, Mercaptopyruvate, Cyc-Gly, Carbomoyl Phosphate, Fumarate, Citrolline, Malate, L-Cysteine, L-Ornithine} \\
		SCC-2 & \textit{Glycerone Phosphate, Ribose 5-Phosphate, Glyceraldehyde 3-Phosphate, PRPP, D-Xylulose 5-Phosphate, beta-D-Fructose-6-phosphate, Sedoheptulose 7-phosphate, beta-D-Fructose 1,6-bisphosphate} \\
		SCC-3 & \textit{AMP, GDP, DNA, Guanine, Deoxyadenosine, dATP, IMP, dGTP, XMP, Xanthine, Guanosine, dADP, dAMP, dGDP, Inosine, Adenine, Hypoxanthine, Adenosine, GMP, Adenylosuccinate, Xanthosine}\\
\bottomrule
\end{tabular}
\caption{SCCs in the core of the Monkey (\textit{M. Mulatta}) metabolic network.}
\label{tab:monkey-metabolic-scc}
\end{table}

\begin{table}[!ht]
\scriptsize
\centering
\captionsetup{font=footnotesize}
\begin{tabular}{p{0.25\textwidth}p{0.09\textwidth}p{0.09\textwidth}p{0.55\textwidth}} 
\midrule
\midrule
		Name & $\frac{P(v)}{P}$ & $\delta_{C}(v)$ & Description\\
\midrule
\multirow{2}{*}{\parbox{2.5cm}{Planned Parenthood v. Casey (1992)}} & 0.69 & 0.69 & \textit{A ``landmark'' decision on abortion rights.} \\
& & \\
Roe v. Wade (1973) & 0.65 & 0.20 & \textit{A ``landmark'' decision in favor of abortion rights with certain restrictions.}\\
Bigelow v. Virginia (1975) & 0.38 & 0.05 & \textit{A ``landmark'' decision on protecting First Amendment right on advertising, where the advertisement in question was on abortion services.}\\
Harris v. McRae (1980) & 0.55 & 0.03 & \textit{A ``landmark'' decision regarding federal funds restriction on abortion.}\\
\bottomrule
\end{tabular}
\caption{The waist of the SCotUS citation network on Abortion cases. Cases labeled as ``landmarks'' are listed as Historic by the Legal Information Institute at Cornell University.}
\label{tab:core-nodes-abortion}
\end{table}

\begin{table}[!ht]
\scriptsize
\centering
\captionsetup{font=footnotesize}
\begin{tabular}{p{0.27\textwidth}p{0.09\textwidth}p{0.09\textwidth}p{0.55\textwidth}} 
\midrule
\midrule
		Name & $\frac{P(v)}{P}$ & $\delta_{C}(v)$ & Description\\
\midrule
Goldberg v. Kelly (1970) & 0.42 & 0.42 & \textit{A ``landmark'' decision that established the full evidential hearing requirement before termination of welfare benefits.}\\
\multirow{2}{*}{\parbox{3cm}{Allied Structural Steel Co. v. Spannaus (1978)}} & 0.22 & 0.22 &  \textit{A ``landmark'' decision that reinstated pension rights for certain Allied Steel employees.}\\
\multirow{2}{*}{\parbox{3cm}{L.A. Dept. of Water \& Power v. Manhart (1978)}} & 0.16 & 0.11 &  \textit{A ``landmark'' decision that stated discrimination in pension contribution requirement based on sex is unlawful.}\\
\multirow{2}{*}{\parbox{3cm}{US Railroad Retirement Bd. v. Fritz (1980)}} & 0.38 & 0.09 &  \textit{A ``landmark'' decision that reinstated pension rights for certain US Railroad employees.}\\
Johnson v. Robison (1974) & 0.21 & 0.03 & \textit{A decision that retained certain benefits for combat veterans.}\\ 
\multirow{2}{*}{\parbox{3cm}{Hishon v. King \& Spalding (1984)}} & 0.08 & 0.02 &  \textit{A decision regarding benefit discrimination based on sex.}\\
& & \\
Helvering v. Davis (1937) & 0.06 & 0.02 & \textit{A ``landmark'' decision defending the constitutional validity of the Social Security Act.}\\
\multirow{2}{*}{\parbox{3cm}{Nollan v. California Coastal Com. (1987)}} & 0.07 & 0.01 & \textit{A decision concerning 5th and 14th amendment for property protection.} \\
\multirow{2}{*}{\parbox{3cm}{United States v. Kokinda (1990)}} & 0.11 & 0.01 & \textit{A decision involving first amendment rights for free speech.} \\
& & \\
\multirow{2}{*}{\parbox{3cm}{Pension Benefit Guar. Corp. v. LTV Corp. (1990)}} & 0.17 & 0.01 & \textit{A decision involving insurance of pension benefits.} \\
& & \\
\multirow{2}{*}{\parbox{3cm}{Plaut v. Spendthrift Farm (1995)}} & 0.11 & 0.01 & \textit{A decision concerning separation of power between legislation and judiciary.} \\
\bottomrule
\end{tabular}
\caption{The waist of the SCotUS citation network on Pension cases. Cases labeled as ``landmarks'' are listed as Historic by the Legal Information Institute at Cornell University.}
\label{tab:core-nodes-pension}
\end{table}

\clearpage

\begin{figure}
    \centering
		\subfigure[OpenSSH call-graph]
    {
        \includegraphics[scale=0.4]{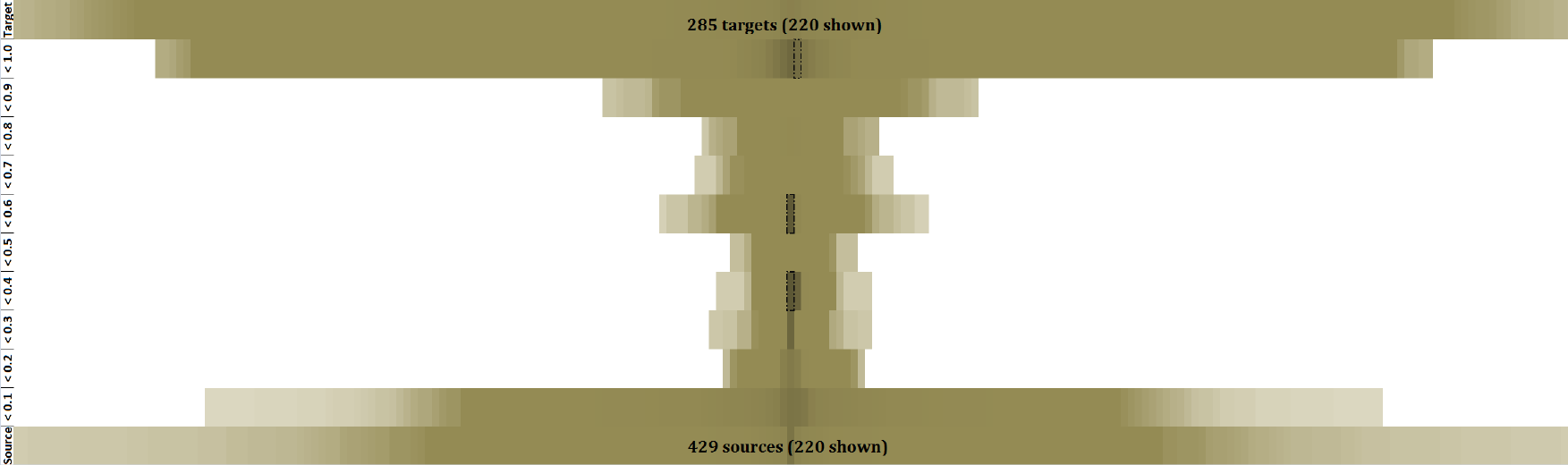}
    }
		\hspace{1mm}
		\subfigure[Apache math call-graph]
    {
        \includegraphics[scale=0.4]{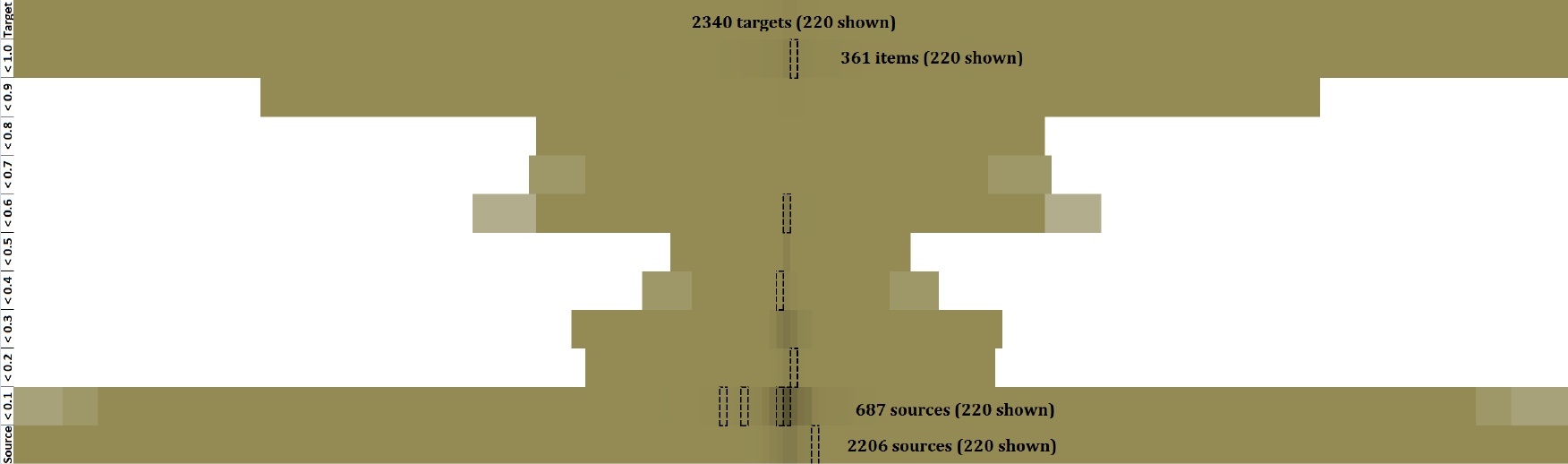}
    }
		\hspace{1mm}
				\subfigure[Monkey metabolic network]
    {
        \includegraphics[scale=0.4]{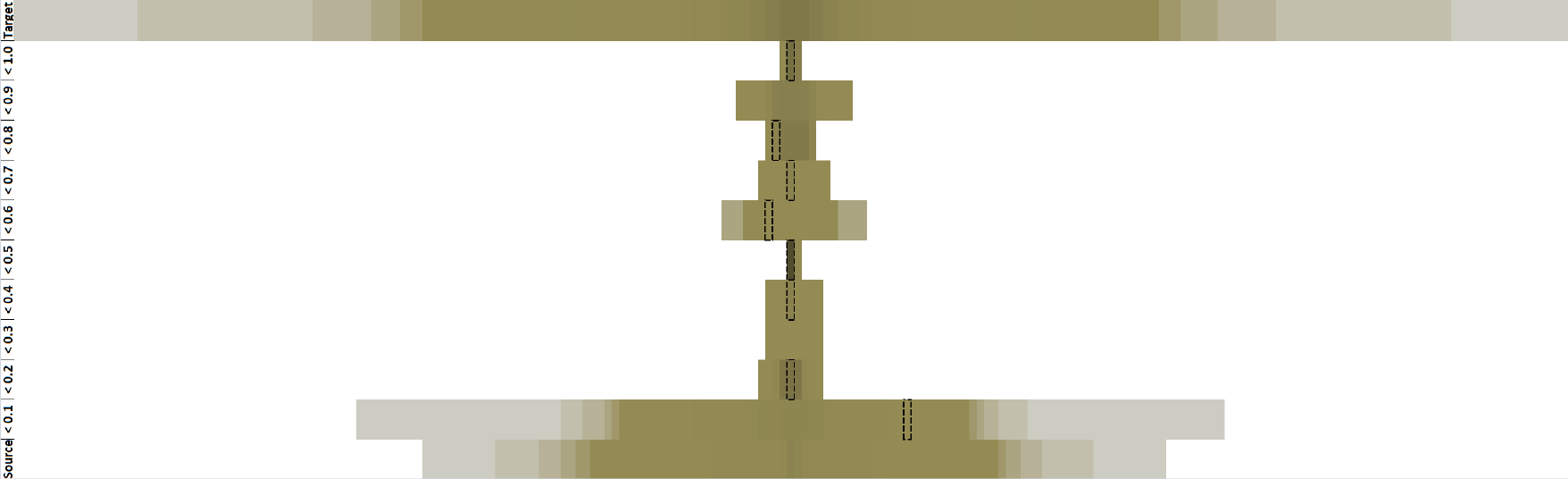}
    }
		\hspace{1mm}
		\subfigure[Abortion case citation network]
    {
        \includegraphics[scale=0.4]{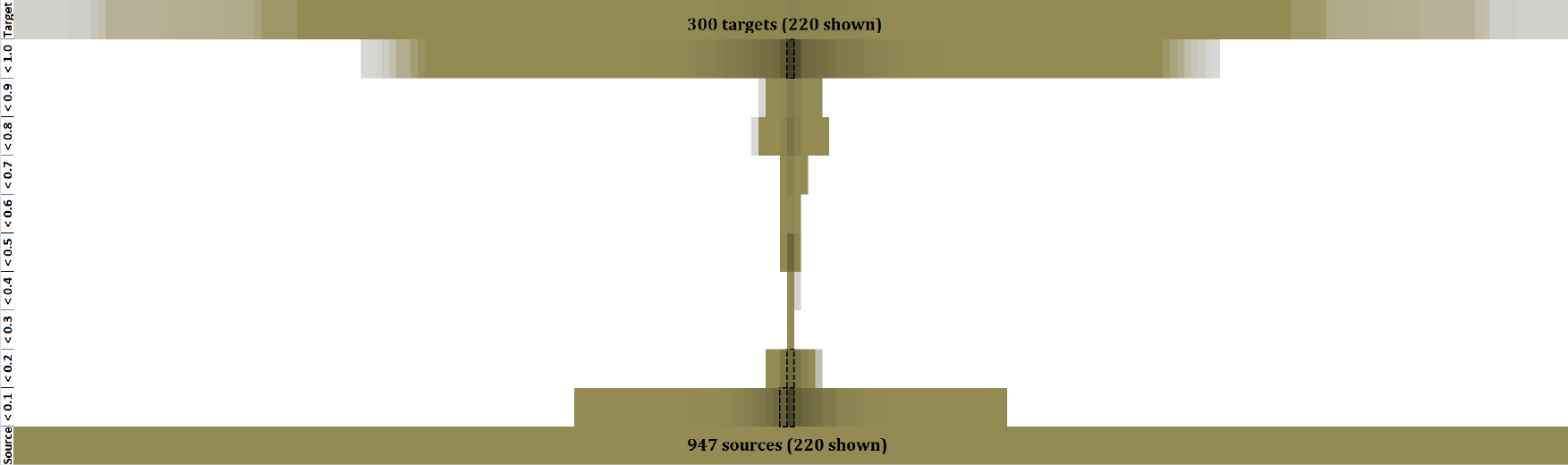}
    }
		\hspace{1mm}
		\subfigure[Pension case citation network]
    {
        \includegraphics[scale=0.4]{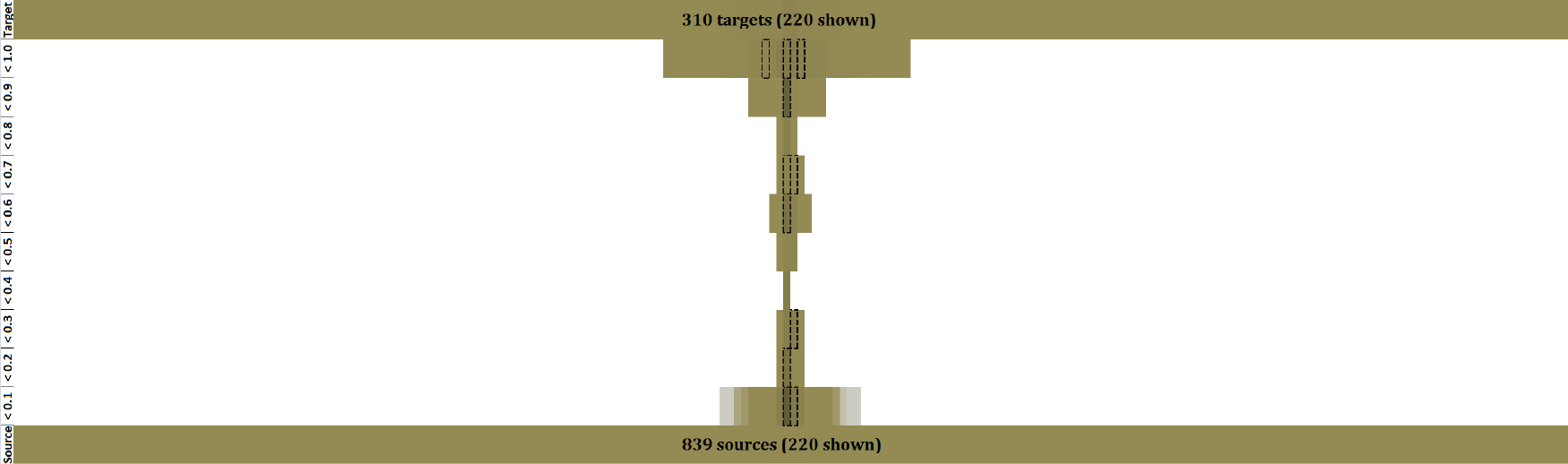}
    }
		\hspace{1mm}
    \caption{Visualizations of the location and path centrality for each network. Please refer
to the caption of Figure~\ref{fig:core-locations-heatmap} for a description of this visualization.}
  \label{fig:heatmap-othernets}
\end{figure}

\clearpage
\begin{figure}
        \centering
				\subfigure[]
				{
					\includegraphics[scale=0.405]{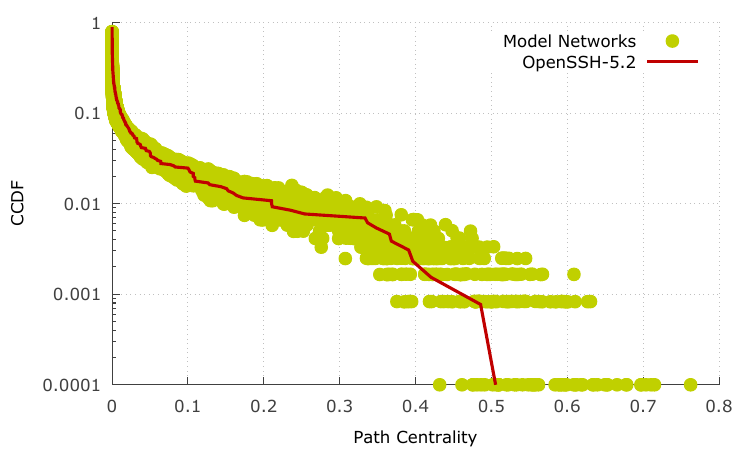}
				}
				\hspace{2mm}
				\subfigure[]
				{
					\includegraphics[scale=0.405]{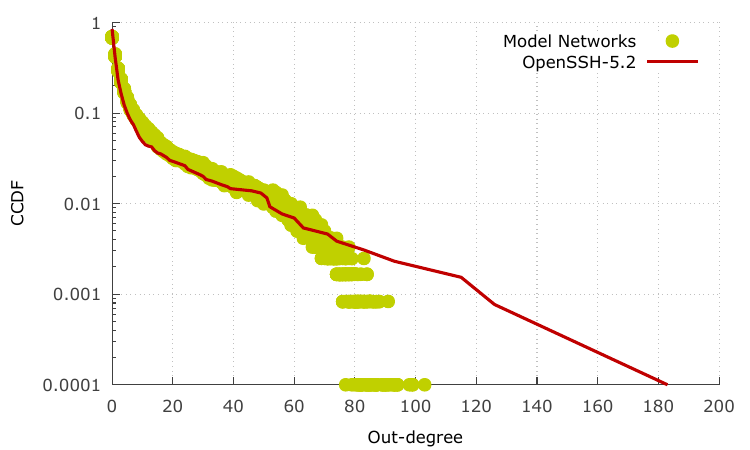}
				}
				
				\subfigure[]
				{
					\includegraphics[scale=0.405]{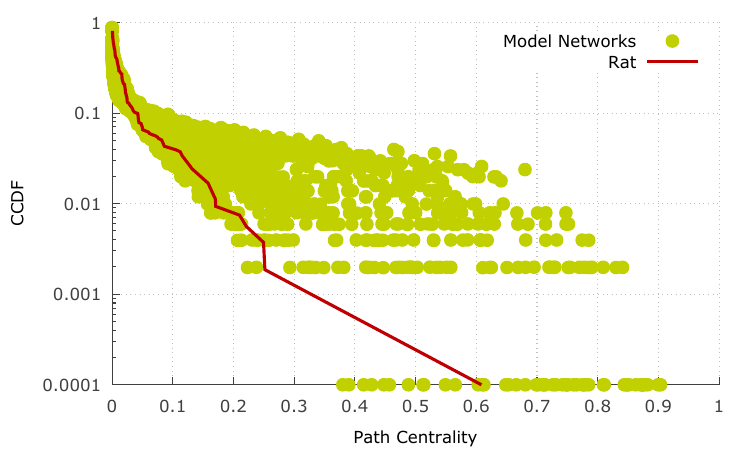}
				}
				\hspace{2mm}
				\subfigure[]
				{
					\includegraphics[scale=0.405]{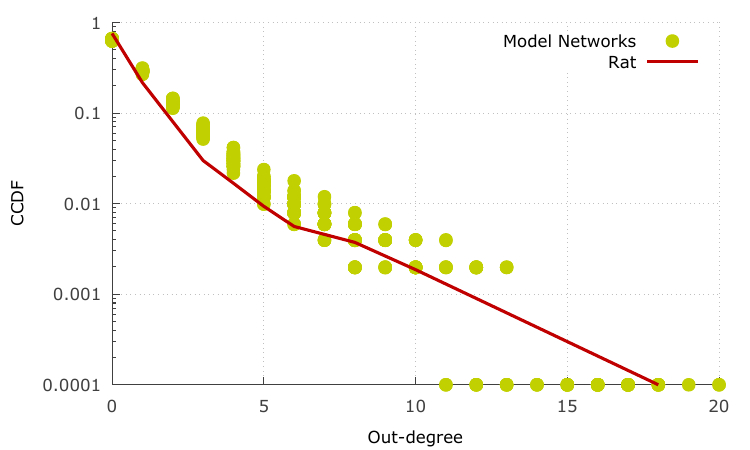}
				}
				
				\subfigure[]
				{
					\includegraphics[scale=0.405]{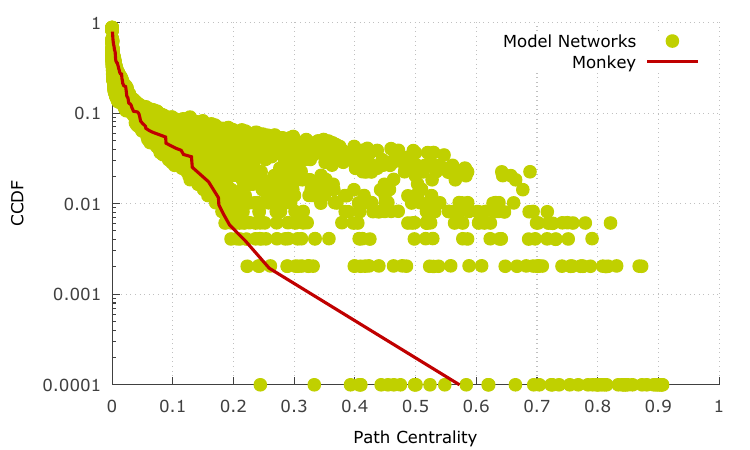}
				}
				\hspace{2mm}
				\subfigure[]
				{
					\includegraphics[scale=0.405]{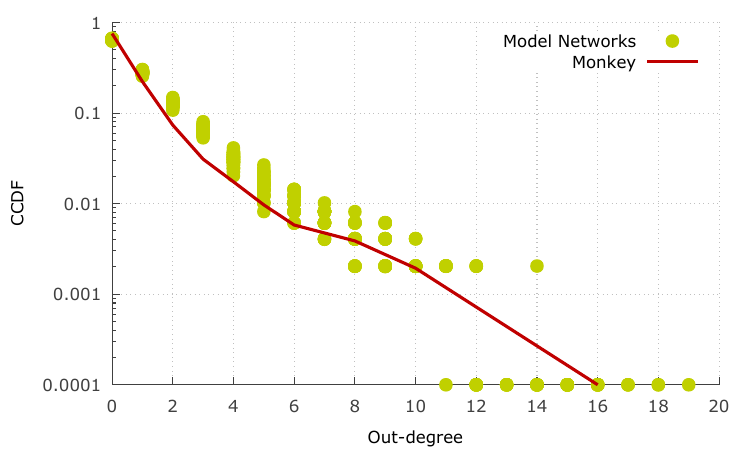}
				}

				\subfigure[]
				{
					\includegraphics[scale=0.405]{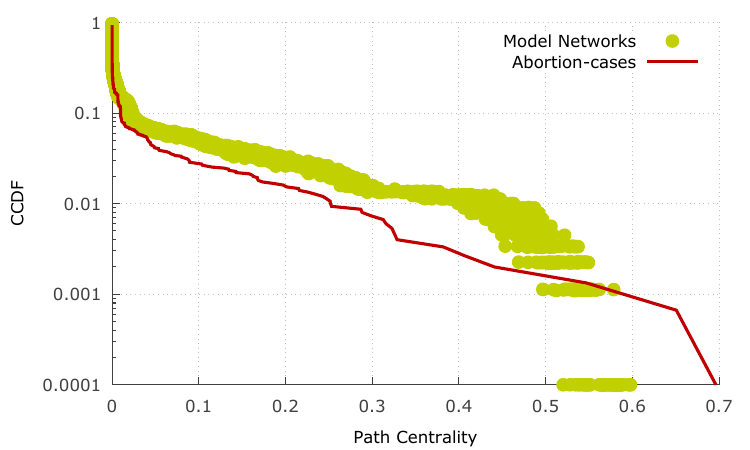}
				}
				\hspace{2mm}
				\subfigure[]
				{
					\includegraphics[scale=0.405]{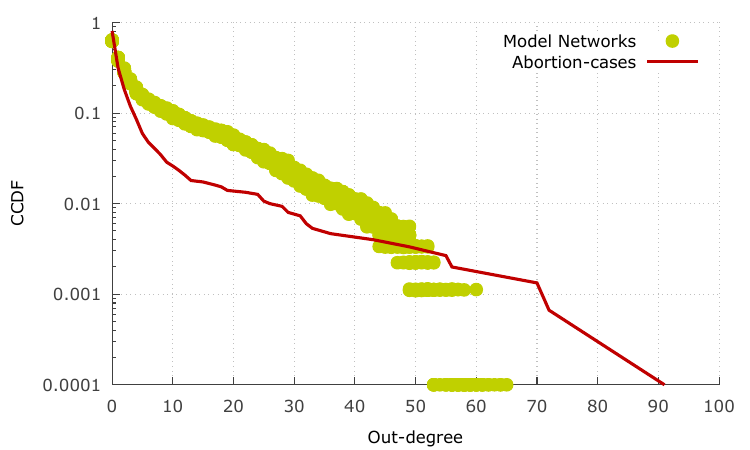}
				}
				
				\subfigure[]
				{
					\includegraphics[scale=0.40]{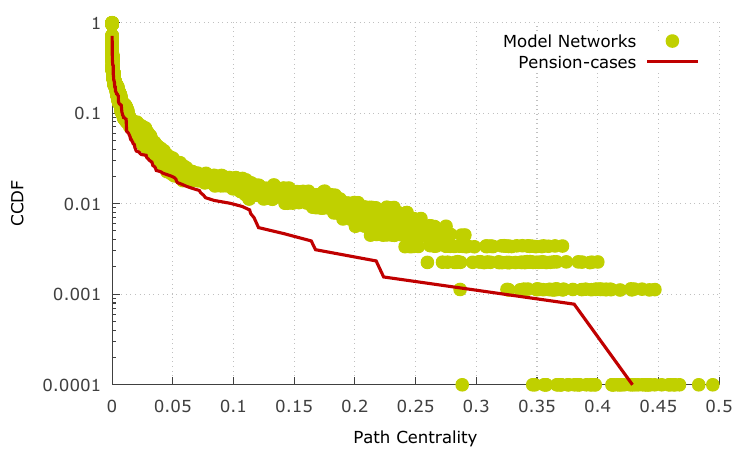}
				}
				\hspace{2mm}
				\subfigure[]
				{
					\includegraphics[scale=0.41]{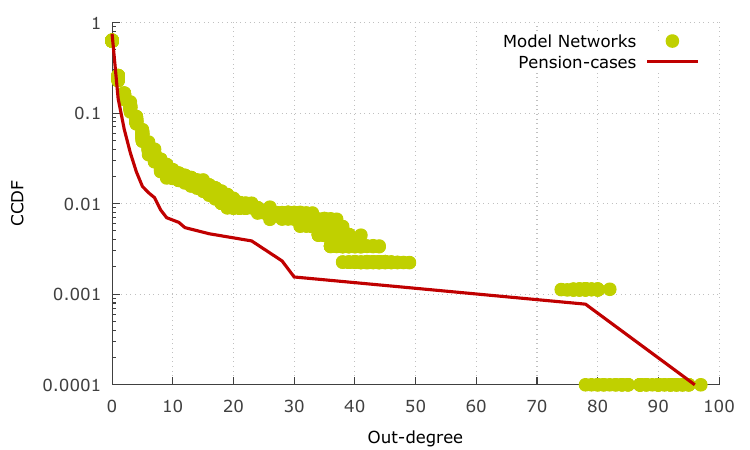}
				}     

        \caption{Comparison of path centrality and out-degree distributions between some real dependency networks and the corresponding synthetic networks generated by the RP-model.}
    \label{fig:real-model-distribution}
\end{figure}
\clearpage

\end{document}